\documentclass{article}
\usepackage[utf8]{inputenc}
\usepackage[english]{babel}

\usepackage[margin = 1in]{geometry}

\usepackage{amsmath}
\usepackage{amsthm}
\usepackage{thmtools}
\usepackage{graphicx}
\usepackage[colorlinks=true, allcolors=lizDkPur]{hyperref}
\usepackage[nameinlink,capitalise]{cleveref}

\usepackage[style=numeric, sorting=none, backend=biber,maxnames=10]{biblatex}
\usepackage{amssymb}
\usepackage{amsfonts}
\usepackage{bbm}
\usepackage[dvipsnames]{xcolor}
\usepackage{tikz}
\usetikzlibrary{positioning}
\usetikzlibrary{shapes.geometric}
\usepackage[normalem]{ulem}

\usepackage{pgfplots}
\pgfplotsset{compat=1.18}
\usepackage{adjustbox}
\usepackage{tikz}
\usepackage{tikz-cd}
\tikzcdset{scale cd/.style={every label/.append style={scale=#1},
    cells={nodes={scale=#1}}}}
\usepackage{quiver}
\usepackage{authblk}
\usepackage[safe]{tipa}
\usepackage{float}
\usepackage[pdf]{pstricks}
\usepackage{pgf}
\usepackage{scalerel}
\usepackage{algpseudocode}
\usepackage{graphbox}
\usepackage{multirow}
\usepackage{thm-restate}
\usepackage{thmtools}
\usepackage{mathtools}
\usepackage{hyperref}
\usepackage{booktabs}

\newtheorem{theorem}{Theorem}[section]
\newtheorem{corollary}[theorem]{Corollary}
\newtheorem{proposition}[theorem]{Proposition}

\newtheorem{lemma}[theorem]{Lemma}
\theoremstyle{definition}
\newtheorem{definition}[theorem]{Definition}
\newtheorem{example}[theorem]{Example}
\newtheorem{remark}[theorem]{Remark}

\newenvironment{examplet}
{\begin{example}
}
{ 
$\hfill \triangleleft$
\end{example} 
}

\definecolor{ltgray}{gray}{0.9}
\definecolor{gray}{rgb}{0.95,0.95,0.96}
\definecolor{dkgray}{rgb}{0.7,0.7, 0.735}
\definecolor{ltblue}{rgb}{0.55,0.55, 0.95}
\definecolor{ltGreen}{rgb}{0.25,0.65, 0.25}
\definecolor{dkgreen}{RGB}{0, 100, 0}
\definecolor{dkred}{rgb}{0.75,0.0, 0.0}
\definecolor{ltred}{rgb}{0.95,0.95, 0.85}
\definecolor{utahRed}{rgb}{.8, 0, 0}
\definecolor{oregonGreen}{rgb}{0, .41, .163}
\definecolor{albanyPurple}{rgb}{0.4, 0, .55}
\definecolor{figcolor}{RGB}{234, 181, 134} 
\definecolor{msugreen}{RGB}{24, 69, 59}
\definecolor{osured}{RGB}{187, 0, 0}
\definecolor{osugray}{RGB}{102, 102, 102}
\definecolor{msuAndosu}{RGB}{97,38,32}

\definecolor{lizRed}{HTML}{D94E3D}
\definecolor{lizDkPur}{HTML}{743B75}
\definecolor{lizLtPur}{HTML}{B87BBC}
\definecolor{lizBlue}{HTML}{26B0E7}
\definecolor{lizYellow}{HTML}{F3A61B}

\hypersetup{
    colorlinks=true,
    linkcolor=lizDkPur,
    filecolor=lizBlue,      
    urlcolor=lizDkPur,
    citecolor = lizBlue,
    pdftitle={Canopies},
    }
    
\usepackage{cleveref}

\crefname{figure}{Figure}{Figures}
\crefname{section}{Section}{Sections}
\crefname{equation}{Equation}{Equations}
\crefname{remark}{Remark}{Remarks}
\crefname{theorem}{Theorem}{Theorems}
\crefname{definition}{Definition}{Definitions}
\crefname{lemma}{Lemma}{Lemmas}
\crefname{proposition}{Proposition}{Propositions}
\crefname{table}{Table}{Tables}

\newcommand{\C}{\mathbb{C}}

\newcommand{\cE}{\mathcal{E}}

\newcommand{\N}{\mathbb{N}}
\newcommand{\R}{\mathbb{R}}
\renewcommand{\S}{\mathbb{S}}

\newcommand{\Z}{{\mathbb Z}}

\newcommand{\cD}{\mathcal{D}}

\newcommand{\Func}{\mathrm{Func}}

\DeclareMathOperator{\cl}{\mathrm{cl}}
\newcommand{\Dgm}{\mathrm{Dgm}}
\newcommand{\DgmA}{\mathrm{dgm_A}}
\newcommand{\DgmAsp}{\mathrm{DGM_A}} 
\newcommand{\DgmD}{\mathrm{dgm_D}}
\newcommand{\DgmDsp}{\mathrm{DGM_D}} 
\newcommand{\e}{\varepsilon}
\renewcommand{\phi}{\varphi}
\newcommand{\inv}{^{-1}}

\newcommand{\id}{\mathrm{Id}}
\DeclareMathOperator{\Id}{\id}

\usepackage{xargs}

\newcommand{\mycomment}[1]{}

\newcommand{\define}[1]{{\bf \boldmath{#1}}}
\newcommand{\intS}[2][n]{\mathbb{I}^{#1}\left[#2\right)}

\title{Canopies: A Generalization of Vines and Vineyards for Parameterized Persistence}
\date{}
\usepackage{authblk}

\author[1]{Barbara Giunti}
\author[2]{Elizabeth Munch}
\affil[1]{University at Albany -- SUNY}
\affil[2]{Michigan State University}

\begin{document}
\maketitle

\begin{abstract}
In this paper, we provide a new construction for studying parameterized persistence, called a \emph{canopy}. 
We give two versions of this construction: the A-canopy, retaining all information about points on the diagonal of the persistence diagram; and the D-canopy, encoding the information of the ``standard" persistence diagram. 
We do this by making a simple but major modification in the persistence bundle representation information: namely, rather than tracking a point in the persistence diagram, we instead track some choice of pairs of simplices that created said point. 
This viewpoint is a combinatorial version of tracking the chain complex information rather than just the output of persistence. 
We show how to construct the canopies from any filtered filtration function, proving, using the algebraic structure of filtered chain complexes, that different choices of pairs result in homeomorphic structures. 
Finally, we showcase the power of our approach by using canopies to define \emph{vines} even in the presence of points with multiplicity; to discuss monodromy; and to obtain some immediate results linking non-trivial monodromy in the persistent homology transform with the existence of \emph{non-Hausdorff points} in the canopy.
\end{abstract}

\section{Introduction}

Topological Data Analysis (TDA), and its spearhead, Persistent Homology (PH), have been proven useful in hundreds of applications \cite{donut}. 
The vast majority of persistence applications assume a single input filtration function $f\colon K \to \R$ where $f(\sigma) \leq f(\tau)$ for any faces $\sigma \subseteq \tau$, then study the homology of a filtration of sublevelsets to get a single persistence diagram \cite{DiFabio2013}.
However, input data often has more than one parameter involved which cannot be studied with standard persistence without making additional choices or restrictions.
In this paper, we focus on the case of continuous input function data $f \colon K \times B \to \R$ defined on a finite simplicial complex $K$, where for any fixed $p \in B$, $f_p\coloneqq f(\cdot, p)\colon K \to \R$ is a filtration function; this $f$ is called a \emph{fibered filtration function} (a term coined by \cite{Hickok2026}). 
Taking the sublevelset persistence for each $f_p$ results in a continuous function $F\colon B \to \Dgm$ where $\Dgm$ is the space of persistence diagrams, which we call \emph{parameterized persistence}. 
Handling data in this form requires more structure than standard persistence since the diagrams are evolving with the changing parameter in $B$.  

By way of disambiguation, we note a distinction here between the \emph{parameterized persistence} we study, and \emph{multiparameter/multidimensional persistence} \cite{Carlsson2009,botnan_lesnick2023} which is also extensively studied in the literature. 
In the multi-parameter persistence setting, an input function is given in the form $f\colon K \to \R^n$ and then homology of sublevelsets $\{\sigma \in K \mid f_i(\sigma) \leq a_i,\; i \in \{1,\ldots,n\} \}$ leads to a persistence module indexed over $\R^n$ rather than $\R$ like in standard persistence. 
The closest relation to the structure we study are cases where the multiparameter persistence module is studied via linear slices (e.g.~\cite{Lesnick2022,CaiKimMemoliWang2021,FernandesOudotPetit2025,BerkoukPetit2022,LesnickWright2025,Piekenbrock2024}) which does result in parameterized persistence.

The beginning of studying data in the parameterized persistence setting was focused on a 1-parameter family of persistence diagrams, termed vineyards \cite{CohenSteiner2006,Morozov2008,Munch2013}. 
In this context $B = [a,b] \subset \R$ is some interval, and the name came from visualizing the evolving persistence diagram $F\colon [a,b] \to \Dgm$ in a 3rd dimension, making the trajectories of the off-diagonal points coming in and out of the diagonal hyperplane look like vines. 
In \cite{CohenSteiner2006}, computation of a vineyard was done by maintaining the reduced boundary matrix over the course of swapping an adjacent in the total order of simplices, at a time of $O(n)$ per transposition. 
Further improvements were made by \cite{Piekenbrock2024} building on ideas from \cite{Busaryev2010}, where knowledge of the upcoming transpositions can be used to speed computation time. 
Vineyards have been used in the context of applications, including 
music \cite{Bergomi2020},
dynamical systems \cite{Algar2021},
neuroscience \cite{Yoo2016,Salch2021},
geospatila data \cite{Hickok2022}, 
plant biology \cite{migicovsky_rootstock_2019}, and fluid dynamics \cite{Soler2018}.
In addition, \cite{Turner2023} is an intensive study of the algebraic structure of vineyards but only away from points of multiplicity. 
Further, \cite{Munch2015,Tchitchek2025} study summary diagrams akin to means in the setting of evolving persistence.

More recently, a larger cohort of input structures have emerged which fit in this parameterized persistence framework for $B$ not simply an interval. 
The most notable of these is the persistent homology transform (PHT) \cite{Turner2014}, which in its basic form assumes as input a simplicial complex $K$ embedded in $\R^d$. 
In this case, we have a fibered filtration function called the \emph{directional transform} defined as $f\colon K \times \S^{d-1} \to \R$.  
This fibered filtration comes from a function defined on the vertex set by $f(v, \omega) = \langle x_v, \omega\rangle$ where $x_v$ are the embedded coordinates of vertex $v$ and $\langle \cdot, \cdot \rangle$ is the standard inner product, and then extending the function to all of $K$ using the lower star filtration. 
Of course, a different choice of function on the embedded simplicial complex results in other types of fibered filtration functions. 
For example, the \emph{distance transform} is a fibered filtration function $f \colon K \times \R^n \to \R$ for $K$ embedded in $\R^n$, which is extended from the function on the vertex set given by $f(x_v,y) = d(x_v,y)$. 
This distance to a point transform can be extended to studying a \emph{distance to flat transform} given similarly by the distance to some line, plane, or hyperplane \cite{Turkes2022,Onus2024,Chambers2025}. 
Another fibered filtration function arises from studying a point cloud moving around called a \emph{dynamic metric space} \cite{Kim2020a,Kim2020} and the function $f\colon K \times \R \to \R_{\geq0}$ is defined for $K$ the complete simplicial complex on the same number of vertices as points. 
These settings naturally arise in the context of flocking and swarming applications. 

An issue that permeates the use of vineyards in applications, and parameterized persistence more generally, is that of a lack of unique choice of pairing for the vines that can occur at points of multiplicity. 
Indeed, to these authors' knowledge, there is no existing, well-defined definition of the word ``vine'' in the literature, despite the term colloquially being used for the connected trajectories from the points of the vineyard. 
This is indeed well defined away from points of multiplicity; however, if there comes a time when an instant of the path $F(p)$ has a point of multiplicity in the persistence diagram, one must make a choice as to how to continue the two trajectories, and this choice is not always unique. 

In order to avoid the issues of truly tracking the proper matching of the vines, avoidance tactics have arisen in practice to handle vineyard data.  
For example, rather than actually tracking the persistence points via bases, a common tactic is to simply work with pointwise calculations for each $F(p)$ such as the maximum persistence; e.g.~\cite{Tymochko2020a,Myers2019}. 
In other cases \cite{giusti2025signatures,garcia2021tracking,Soler2019}, the full time series of persistence diagrams are used, and any trajectory matching is done via the bottleneck or Wasserstein distance matching. 
As noted before, this is reasonable away from points of multiplicity, but has no unique choice available when these points are reached. 

A more complex variant of the pointwise computation trick is to replace the time series of persistence diagrams with a simplified topological descriptor. 
If we replace a diagram with a Betti curve, the result is a CROCKER plot \cite{Topaz2015,Xian2022}, where now the entire construction can be visualized as a matrix for analysis and has been used for studying aphid swarming \cite{Ulmer2018} and dynamical systems \cite{Guezel2022,Tanweer2024}.
If we replace the homology computation with the Euler characteristic, the resulting simplification again can be visualized  more easily in a matrix. 
For example, the analogue of the PHT is the Euler Characteristic Transform \cite{Turner2014,Munch2024} which has found extensive use in applications due to its speed of computation \cite{Wang2021,Amezquita2021,CisewskiKehe2023,Yahiaoui2026,Ayub2026}.
If we restrict the dimension to studying the 0-dimensional persistence of dynamic metric spaces, the resulting structure can be studied via a formigram to focus only on the changing clustering \cite{Kim2018}. 

Another avoidance tactic is doing zigzag persistence \cite{Carlsson2010}; however, choices must be made in order to compute zigzag persistence from data arising in the form $f \colon K \times [a,b] \to \R$. 
For example, we can choose a collection of values $a=a_0<a_1<\cdots<a_k=b$ and a fixed sublevelset parameter $K_{s,a_i} = f_{a_i} \inv(-\infty,s]$. 
The upshot of this method is that a single persistence diagram is returned that can give information about tracking of classes across the parameterization time $[a,b]$. 
This was utilized, for example, in \cite{Tymochko2020} to track an evolving dynamical system and \cite{Myers2023} to track information across a temporal network. 
While available code has recently become much faster \cite{Dey2021,Dey2022}, the downside of this framework is that it is still not a full representation of the available information in the parameterized filtration function, as is our goal here.

Thus, the main focus of our work is based in the recent seminal work of Hickok \cite{Hickok2026,Hickok2022Computation}, termed persistence bundles. 
In that work, she defined and studied fibered filtration functions, and proved that for a generic perturbation of the input, the base space $B$ has a stratification where the ordering of the simplices induced by the filtration function is constant. 
She also constructed what she called \emph{persistence bundles}\footnote{We emphasize that this is distinct from other uses of bundles in the TDA literature, where one is finding coordinates for the dataset itself rather than using bundles to study the algebraic structures, e.g.~\cite{Turow2025}.}, which consist of a fiber-bundle-like structure $\pi \colon E \to B$, where the points of $E$ are of the form $(p,b,d)$ for each point in the persistence diagram of function $f_p$ at birth, death coordinates $(b,d)$. 
However, to truly be a fiber bundle \cite{Husemoller1994}, every point $p \in B$ has a neighborhood $U \ni p$ so that $\pi\inv(U)$ is homeomorphic to $U \times F$ for some fiber $F$. 
In the case of persistence bundles, points entering the diagonal mean the inverse image of a neighborhood might not have the same fiber everywhere, so persistence bundles are not truly fiber bundles (see also \cite[Remark 3.3]{Hickok2026}).
This paper arose from studying exactly what that failure meant, how to control it, and what structure we can actually promise. 
Other work has also studied parameterized persistence, in specific or general settings, but always with a focus on tracking points as being points encoded via (birth, death) coordinates \cite{Fasy2024}. 
We take a different tactic by encoding points by a choice of pairing for tracking purposes.

Of additional interest is that as these parameterized persistence settings have been studied, an increasing library of examples have been found showing that the parameterized persistence with proper tracking of points can result in very interesting and complex structure. 
This behavior has been collectively termed \emph{monodromy}\footnote{Again for disambiguation, we note that we study monodromy arising in the algebraic structures rather than monodromy in the topological space, e.g.~\cite{Burghelea2013,Burghelea2017}.}, where following a closed path in the base space can result in tracked points in the persistence diagram swapping locations. 
Known to these authors, the first mention fitting this context of monodromy in the literature was \cite{Cerri2013}, where a 1-parameter path of lines in a 2-parameter persistence module exhibited this behavior. 
In \cite{Arya2024}, a simple spiral example was shown to exhibit monodromy in the PHT, and further it was proved that for a restricted case of embedded shapes, the PHT does not exhibit monodromy. 
However, we also know that the resulting structure can be as complex as desired since \cite{Chambers2025} proved that any knot can be generated by the paths of points in parameterized persistence defined by the distance transform (specifically, the \emph{radial transform}).

In parallel, there is a growing body of literature on basis matching, decomposition, and the underlying structure encoded at the level of filtered chain complexes \cite{Chacholski2020,Chacholski2021,Chacholski2026}.
The fact that the information enclosed on the ``diagonal'', i.e., the bars of length zero in the barcode that are technically not present at the homological level, is necessary in many settings, prominently for the inverse problem \cite{Fasy2018,Fasy2025,Fasy2022}, is not new. 
It is also known that it encodes richer information than the usual barcode \cite{Chacholski2020,Chacholski2021}. 
Nevertheless, in the literature the focus seems to mostly be on matching points in the persistence diagram, rather than the generating pairs of simplices, perhaps because there are only a few and recent stability results that take into account points on the diagonal \cite{Memoli2023}. 
Therefore, the full potential of the information encoded at the chain complex level is still untapped.
While care needs to be taken when using directly pairings of simplices, since they are in general not uniquely defined (rather, there are several different possible representatives), it is possible to compute the augmented persistent diagram from any fixed pairing, as argued in \cite{Chacholski2020}.
Moreover, as hinted in \cite{Chacholski2020} and made explicit here, the resulting persistence diagram is independent of the choice of pairing, and we have precise ways to move from one pairing to another preserving all the needed structure.

\paragraph{Our contribution}
In this paper, we provide a new construction for studying parameterized persistence, called a \emph{canopy}. 
We give two versions of this construction, the A-canopy (\cref{def:A-canopy}), which can be likened to augmented diagrams retaining all information about points on the diagonal of the persistence diagram; and the D-canopy (\cref{def:D-canopy}), which represents that information which would still be retained in the ``standard" persistence diagram after the contents of the diagonal are forgotten. 
We do this by making a simple but major modification in the persistence bundle representation information: namely rather than tracking a point in $E$ by points in the persistence diagram $(p,\text{birth},\text{death}) \in B \times \R \times \R$, we instead track them by some choice of pairs of simplices which created those points under some choice of total ordering of the simplicial complex compatible with the filtration function, $(p,\sigma_b, \sigma_d)$. 
This viewpoint is a combinatorial version of tracking the chain complex information rather than just the output of persistence. 
We solve the problem of lack of unique matching by defining a topology on $E$ that arises from tracking points via the algebraic structure presented in \cite{Chacholski2020}, specifically from the pairings of simplices that generate the indecomposables of the filtered chain complexes at each base point. 
While a pairing is not, in general, uniquely defined (i.e., for the same filtered function there could be different pairings generating the indecomposables), these indecomposables are isomorphic, and thus there is a homeomorphism between the topologies of any two different pairing choices.
This construction of an A-canopy has a similar flavor to étale spaces of sheaves, or perhaps to branched coverings; however our result is that pairing issues appear as non-Hausdorff points in the space which is not a standard assumption for those constructions.
We note that it is also more broadly applicable than the cellular sheaf construction of \cite[Sec.~6]{Hickok2026} since this does not require a stratification on the base space to be defined.

In \cref{thm:A-bundle-decomp}, we show that every fibered filtration function has an A-canopy, and this construction is unique up to isomorphism (\cref{thm:A-bundle-unique}). 
We show that away from the non-Hausdorff points, the A-canopy for an input fibered filtration function is a true fiber bundle (\cref{thm:A-bundle-decomp-generic}). 
We also give a similar structure theorem for the D-canopy that preserves aspects of the fiber bundle structure where possible; namely, away from the non-Hausdorff points (\cref{thm:D-bundle-decomp}). 
This provides an answer to a question of Turner \cite[Sec.~8]{Turner2023}, although perhaps not as satisfying as one might wish. 
Indeed, it says that rather than finding a way to ``properly'' match points in the diagram when there is ambiguity, instead the topologically correct choice according to the construction is to make no choice at all, and consider the two paths as one. 
However, the construction allows for more control: even in the presence of these issues, we can define and study A-canopies, and some points of multiplicity do not result in non-Hausdorff issues.
For these points, there is a canonical choice of where to go when two paths cross, thus allowing us to not throw the baby out with the bathwater.
We call the non-Hausdorff points with no canonical choice \textit{structural}, as they are exactly where all desired structure in A- and D-canopies break; as opposed to \textit{coincidental} points of multiplicity which do not require additional care.
We conclude the paper with immediate implications arising from the A- and D-canopy constructions in \cref{sec:implications}.

\section{Setting}
\label{sec:Background}

In this section, we set up the stage for the constructions in the paper, focusing on bundles and the required algebraic aspects of persistent homology.

\subsection{Hausdorff points and spaces}

While the definition of a Hausdorff space is well-established, here we need a specification of it for single points:

\begin{definition}
\label{def:non-Haus-point}
Fix a topological space $X$.
A point $x \in X$ is a \define{Hausdorff point} if for all $y \in X$, there are open sets $U \ni x$, $V \ni y$ such that $U \cap V = \emptyset$.
Otherwise $x \in X$ is called a \define{non-Hausdorff point}. 
\end{definition}

While the definition of a non-Hausdorff \textit{point} is not standard, this definition is set up so that if a space has no non-Hausdorff points, then it is a Hausdorff \textit{space}. 

\subsection{Bundles and Fiber Bundles} 
\label{ssec:bundles}

Here we give the basic definitions of bundles and fiber bundles since we need only the basics for our construction; we direct the interested reader to \cite{Husemoller1994,Michor2008,Hatcher2002} for additional details.

\begin{definition}
    A \define{bundle} is a triple $(E,\pi,B)$ where $\pi\colon E \to B$ is a map. 
    The space $E$ is called the \define{total space}; $B$ is called the \define{base space}; and the map $\pi$ is called the projection of the bundle. 
    For each $p \in B$, the space $\pi \inv(p)$ is the \define{fiber} of the bundle over $p$.
\end{definition}

\begin{definition}
    A \define{fiber bundle} is a 4-tuple $(E,B,\pi,F)$ where $E$, $B$, and $F$ are topological spaces, and $\pi\colon E \to B$ is a continuous surjection. 
    We require that for every $p \in B$, there is an open neighborhood $p \in U \subset B$ (called a \define{trivializing neighborhood}) such that there is a homeomorphism $\pi$ such that the diagram 
    \begin{equation*}
    \begin{tikzcd}
    \pi \inv(U) \ar[r,"\phi"] \ar[d,"\pi"] 
    & U \times F \ar[dl,"{\text{proj}_1}"] \\ 
    U
    \end{tikzcd}
    \end{equation*}
    commutes.
\end{definition}

We note that our fiber will often be a set of $N$ elements with the discrete topology; we denote this throughout by $[N]\coloneqq \{ 1,\ldots,N\}$.
In the case of a fiber bundle with a discrete fiber like what we study, the result is a \emph{covering space} \cite[Ex.~4.42]{Hatcher2002}, however we use the term \emph{fiber bundle} in this paper to emphasize the connection with persistence bundles. 

As this will be useful later, we recall some basic results on covering spaces of spheres. 
First, by the classification theorem for covering spaces (see e.g.~\cite[Thm.~1.38]{Hatcher2002}), there is a bijection between the set of basepoint preserving isomorphism classes of path-connected covering spaces and the set of subgroups of $\pi_1(X,x_0)$. 
This means that for an $n$-dimensional sphere $S^n$ with $n\geq 2$ where $\pi_1(S^n) = 0$, the only path-connected cover is the trivial cover $\Id \colon S^n \to S^n$. 
For the covering space of a circle, $S^1$, the only options are the universal cover, $\pi \colon \R \to \S^1$ defined by $\pi(t) = (\cos t, \sin t)$; or the $n$-sheeted cover $\pi \colon \S^1 \to \S^1$ given by $z \mapsto z^n$ viewing $z \in \C$ with $|z|=1$.

\subsection{Pairings and persistence diagrams}

We will build two types of persistence diagrams from a given input: one which tracks all pairings and one which only tracks the pairings with non-zero persistence (the off-diagonal points).
Let $\overline{\R} = \R \cup \{ \infty\}$ be the extended real line. 

\begin{definition}
\label{def:AllTheHs}
    The \define{upper half space} and its interior are given by 
    \[H = \{(x,y) \in \R \times \overline\R\mid x\leq y\} \qquad \text{and} \qquad 
    \mathring{H} = \{(x,y)\in \R \times \overline\R \mid x < y \},
    \]
    respectively. 
    The subset  $\{(x,x) \in \R^2\}\subset H$ is called the \emph{diagonal} and is denoted by $\Delta$. 
    We treat $\Delta$ as a single element, with topology induced by the quotient $H \rightarrow H/\sim$ for $(x,x) \sim (y,y)$ $\forall x,y \in \R$.
    We define $H'$ to be 
    \[H' =  \mathring{H} \cup \{ \Delta\}\, .\]  
\end{definition}

The default assumption in the literature is to have persistence diagrams without points in the diagonal. 
However, persistence diagrams tracking diagonal information have become increasingly common; these are sometimes called 
augmented diagrams (such as in \cite{Fasy2022,Belton2020,Micka2020}),
verbose diagrams (as in \cite{Fasy2025,Usher2016})
or ephemeral diagrams (such as in \cite{Memoli2023,Berkouk2021,Chazal2016}).
Here we use the term augmented, and with connection to music theory, call the standard persistence diagrams diminished.

\begin{definition}
    An \define{augmented persistence diagram (A-diagram)} is a finite multiset $X \subset H$. 
    A \define{diminished persistence diagram (D-diagram)} is a finite multiset $X \subset \mathring{H}$. 
    The space of augmented persistence diagrams is denoted $\DgmAsp$, and the space of diminished persistence diagrams is denoted $\DgmDsp$.
\end{definition}

Note that because we might have more than one class born and dying at the same pair of function values, we need these sets to be treated with multiplicity. 
Also, any A-diagram can be converted into the relevant D-diagram by simply forgetting the points in $\Delta$; however without additional input information, we can not go the other way.
We next show how to get an augmented and diminished persistence diagram from an input function. 
To do so, we recall that a \define{filtered chain complex} is a functor from $\mathbb{R}$ (see as a category) to the category of non-negative chain complexes over vector spaces, such that all chain maps from $s$ to $t$, for all $s\leq t\in\R$, are monomorphisms.

\begin{definition}
Assume we are given a finite simplicial complex $K$. 
A \define{filtration function} is a function $f\colon K \to \R$ such that $\sigma \subseteq \tau$ implies $f(\sigma) \leq f(\tau)$. 
The \define{sublevelset} of $K$ at value $a$ is the subcomplex $K_a = f\inv(-\infty,a]$. 
The \define{(sublevelset) filtered chain complexes associated with $f$} is the filtered chain complex given by the sublevelset filtration of $K$.
\end{definition}

Note that this definition is setup so that $K_a$ is always a subcomplex of simplicial complex $K$, and also that $K_a \subseteq K_b$ for any $a \leq b$. 
A common method to build a filtration function from a function defined only on the vertices is that of the \define{lower star filtration}, where given $f\colon V \to \R$, we abuse notation to define $f \colon K \to \R$ by $f(\sigma) = \max\{f(v) \mid v \in \sigma\}$.
It is bookkeeping to check that the lower star filtration is a filtration function for any choice of function on the vertex set.

We will construct the sublevelset persistence diagram from this input setup as follows. 

\begin{definition}
    A \define{consistent} ordering of $K$ with respect to $f$ is a order $\prec_f$ on the simplices of $K$ such that if $\sigma \subseteq \tau$ or if $f(\sigma) < f(\tau)$, then $\sigma \prec_f \tau$. 
    We call such an order a \define{total consistent order} if it is a total order.
    Moreover, there is a unique \define{minimal consistent order} where every relation $\sigma \prec_f \tau$ implies either $\sigma \subseteq \tau$ or $f(\sigma) \leq f(\tau)$ (or both).
\end{definition}

Note that a consistent \emph{total} ordering for a function is not unique in general; however, uniqueness is achieved if the function values of the simplices are all unique. 
In particular, if $f$ takes different values on every simplex of $K$, then the minimal (and only) consistent order is a total order. 
On the other hand, given a partial minimal consistent order, there could possibly be many different total consistent orders.

\begin{definition}
\label{def:refinement}
    A total order $\prec_T$ is a \define{refinement} of a partial order $\prec_p$ if $a \prec_p b$ implies $a \prec_T b$.
\end{definition}

Given a consistent order (total or partial), we can track the pairings that give rise to the interval modules in the decomposition of the sublevelset persistence module. 
In order to handle infinite classes, we introduce a symbol $\star$ to allow a simplex giving birth to an infinite class being paired to something. 
For this reason, we denote $\overline{K} = K \cup \{\star\}$, with the assumption that $f(\star) = \infty$. 
With this structure, we have the following theorem. 

\begin{theorem}[{\cite[Prop. 3.2 and 3.3]{Chacholski2020}}]\label{thm:simplices_dec}
Fix a consistent ordering $\prec_f$ of $f$. 
Then there exists a subset $P_\prec \subset K \times \overline{K}$ such that the filtered chain complex given by the sublevelset filtration of $\prec_f$ decomposes into a finite direct sum of interval sphere $\intS{f(\sigma), f(\tau)}$ for $(\sigma,\tau)\in P_\prec$, where $f(\star)=\infty$. 
\end{theorem}

In the above theorem, an \define{interval sphere} is an object $\intS{a,b}$ of the form: 
\[
\begin{tikzcd}[column sep={1.1cm}, ampersand replacement = \&]
\arrow[d, phantom, "\scriptstyle 0" description]  
\& \arrow[d, phantom, "\scriptstyle\cdots" description]
\& \arrow[d, phantom, "\scriptstyle b" description]
\& \arrow[d, phantom, "\scriptstyle\cdots" description]
\& \arrow[d, phantom, "\scriptstyle d" description]  
\& \arrow[d, phantom, "\scriptstyle\cdots" description]
\&
\\
0 \arrow[r]\arrow[d]
\& \cdots \arrow[r]
\& 0 \arrow[r]\arrow[d]
\& \cdots\arrow[r]
\& k\arrow[r, "\id"]\arrow[d, "\id"]
\& \cdots
\& \arrow[l, phantom, "\scriptstyle n+1" description]
\\
0 \arrow[r]
\& \cdots\arrow[r]
\& k \arrow[r, "\id"]
\& \cdots \arrow[r, "\id"]
\& k \arrow[r, "\id"]  
\& \cdots
\& \arrow[l, phantom, "\scriptstyle n\quad" description]
\end{tikzcd}
\]
where $k$ is a field, the degree $n$ is either $0$ or a natural number and the extrema of the interval are such that $0\leq b \leq d\leq \infty \in\overline{\R}$. 
All elements in the other degrees are $0$. 
The horizontal direction represents $\overline{\R}$, and the vertical (downward) direction represents the degrees of the chain complexes. 
We remark that what appears to be an abuse of notation for the intervals (in the symbol of interval sphere) is intentional: if $f(\sigma)=f(\tau)$, then the interval is empty, meaning that the corresponding homology class has 0 lifespan, i.e., it is not seen at the homological level, even if it is present at the filtration level.
\medskip

The fact that the interval spheres are the all and only indecomposables of filtered chain complexes has been shown in this formulation in \cite[Theorem 4.2]{Chacholski2021}, but a structure theorem for filtered chain complexes appeared earlier in the literature in different settings \cite{Usher2016,barannikov1994,deSilva2011Dualities,meehan2017}.
\medskip

\cref{thm:simplices_dec} is proved in \cite{Chacholski2020} using the \emph{split conditions}, which we report here for convenience. 
Given a filtration function $f$ and its associated sublevelset filtered chain complex, we fix a total (not necessarily compatible) ordering on its simplices, and denote by $\partial$ its boundary matrix and by $\partial[\sigma,\tau]$ the element of $\partial$ in row $\sigma$ and column $\tau$.
A pair $(\sigma',\tau')\in K\times K$
is said to satisfy the \textbf{split conditions} \hyperlink{SC1}{\bf SC1-2-3} if 
\begin{align*}
\textbf{SC\hypertarget{SC1}{1}: }& \partial[\sigma',\tau']\neq 0; \\
\textbf{SC\hypertarget{SC2}{2}: }& \sigma' \in \mathrm{argmax}  \{f(\sigma) \ \mid \ \partial[\sigma,\tau']\neq 0, \ \sigma\in K\};
\\
\textbf{SC\hypertarget{SC3}{3}: }& \tau' \in \mathrm{argmin} \{f(\tau)\ \mid \ \partial[\tau,\sigma']\neq 0, \ \tau\in K\} \, .
\end{align*}
Moreover, if $\partial[\sigma,\tau]=0$ for all $\sigma,\tau\in K$, any pair $(\sigma,\ast)\in K\times \{\star\}$ is said to satisfy the split condition \textbf{SC\hypertarget{SC$^\star$}{$\star$}}. 
\medskip

If we assume that the simplices in the matrix $\partial$ are ordered by their evaluations under $f$, split conditions correspond to identifying a \emph{unique pivot}, i.e., the row of the lowest non-zero element in a column which is also the leftmost non-zero element of said row \cite{ZomorodianFOCS2000}. 
The pair \emph{(pivot, column-of-the-pivot)} forms a persistence pair. 
The standard barcode algorithm identifies all such pairs by reducing the matrix with specific row and column operations \cite{ZomorodianFOCS2000}, and could be interpreted as finding all pairs satisfying the split conditions iteratively. 
However, the split conditions do not require the matrix to be ordered in any specific way. 
This would make them, in all likelihood, less efficient for practical implementations, but are extremely convenient for theoretical purposes. 
Indeed, they will allow us, among other things, to trace the pairing across different filtration functions and across different orderings for the same filtration function.

\begin{definition}
Given a filtration function $f\colon K \to \R$, a \define{chosen pairing} is a pairing $P_\prec$ that arises from \cref{thm:simplices_dec}.
\end{definition}

In the above pairing $P_\prec$, each simplex of $K$ appears exactly once, but $\star$ may appear multiple times (specifically, it appears once for each class that does not die). 
Further, as recalled here in \cref{lem:total_order_unique_pairing} for self-containment, if $\prec_f$ is a total order, then a chosen pairing $P_{\prec_f}$ is uniquely determined.

\begin{lemma}
\label{lem:total_order_unique_pairing}
Let $f$ be a filtration function and $\prec$ a consistent total order of a simplicial complex $K$ with respect to $f$. 
Then there is only one chosen pairing $P_\prec$.
\end{lemma}

\begin{proof}
If the order $\prec$ is total, then there are only two simplices satisfying the argmax and argmin conditions in \hyperlink{SC3}{\bf SC2-3}. 
Therefore, for every iteration in which we split a pair, the pair $(\sigma',\tau')$ is uniquely determined, and hence so is the pairing $P_\prec$.
\end{proof}

We can now define the augmented and diminished persistence diagrams of a filtration function using a chosen pairing.

\begin{definition}
\label{def:dgms_from_pairing}
    Fix a consistent ordering $\prec_f$ for a given filtration function $f\colon K \to \R$ and a chosen pairing $P_{\prec_f}$. 
    The $p$-dimensional augmented persistence diagram $\DgmA_p(f)$ for a filtration function is the multiset
    \begin{equation*}
        \DgmA_p(f) = \left\{(f(\sigma),f(\tau)) \mid (\sigma, \tau) \in P_{\prec_f}, \, \dim(\sigma)=p\right\}.
    \end{equation*}
    The $p$-dimensional diminished persistence diagram $\DgmD_p(f)$ is the multiset
    \begin{equation*}
        \DgmD_p(f) = \left\{(f(\sigma),f(\tau)) \mid (\sigma, \tau) \in P_{\prec_f}, \, f(\sigma) \neq f(\tau), \, \dim(\sigma)=p \right\}.
    \end{equation*}
\end{definition}

We note that the augmented version is not the ``standard'' notion of persistence since it is tracking points on the diagonal. 
However, we can forget this information to recover the standard version; here we call these diagrams ``diminished'' to emphasize the opposite of ``augmented.''
\medskip

In order to ensure that \cref{def:dgms_from_pairing} is well-posed, we need to prove that the decomposition into interval spheres is independent of the choice of pairing and of the choice of compatible order. 
We begin by showing the former:

\begin{lemma}
\label{dgm_welldef_wrt_chosenpairing}
Let $P_1$ and $P_2$ be two chosen pairings of a filtration function $f$ with consistent ordering $\prec$. 
Then the augmented and diminished persistence diagrams with respect to $P_1$ and $P_2$ coincide.
\end{lemma}

\begin{proof}
By hypothesis and definitions, both $P_1$ and $P_2$ arise from \cref{thm:simplices_dec}. 
This means that, letting $F$ denote the filtered chain complex obtained by the sublevelset filtration along $f$, we have 
\[
F\cong \displaystyle\bigoplus_{(\sigma,\tau)\in P_1}\intS{f(\sigma), f(\tau)} \text{ and } F\cong \displaystyle\bigoplus_{(\sigma',\tau')\in P_2}\intS{f(\sigma'), f(\tau')} \, .
\]
The claim then follows by the uniqueness of the decomposition in \cite[Theorem 4.2]{Chacholski2021}. 
\end{proof}

We note that, for the case of diminished persistence diagrams, the above lemma is also a consequence of \cite{ZomorodianFOCS2000,CrawleyBoevey2015,Oudot2015}.

\begin{lemma}
\label{dgm_welldef_wrt_consistentordering}
Let $\prec_1$ and $\prec_2$ be two consistent orders of a simplicial complex $K$ with respect to a filtration function $f$, and $P_{\prec_1}$ and $P_{\prec_2}$ two chosen pairings of them. 
Then the augmented and diminished persistence diagrams with respect to $P_{\prec_1}$ and $P_{\prec_2}$ coincide.
\end{lemma}

\begin{proof}
It is sufficient to notice that the split conditions \hyperlink{SC1}{\bf SC1-2-3} depend only on the values of $f$, and are independent of the consistent ordering chosen. 
The claim then follows by \cref{dgm_welldef_wrt_chosenpairing}.
\end{proof}

\begin{proposition}
\label{prop:Num_Pairs_Constant}
Let $K$ be a simplicial complex. 
Then there exists $m\in\mathbb{N}\cup\{0\}$ such that $|P_{\prec_f}|=m$ for every filtration function $f$ over $K$. 
Moreover, there exists $m'\in\mathbb{N}\cup\{0\}$ such that $|\{\sigma \in K \mid (\sigma,\star)\in P_{\prec_f}\}|=m'$ for every filtration function $f$ over $K$.
\end{proposition}

This proof hinges on the fact that filtration functions take values in $\R$ and not in $\overline{\R}$. 
Hence, all simplices in $K$ have a finite filtration value and will eventually appear in the filtration.

\begin{proof} 
Let $n_i$ be the number of $i$-simplices in $K$, for $i=0,1,\dots,\dim(K)$. 
Let $\beta_i$ be the $i$th Betti number of $K$, i.e., the rank of the $i$th homology group of $K$, for $i=0,1,\dots,\dim(K)$. 
Then we have that, for all $i=0,1,\dots,\dim(K)$, $\beta_i$-many simplices are paired with $\star$, since they give birth to $\beta_i$-many essential classes. 
All other simplices are paired and each pair generates an interval sphere (see \cite[Paragraphs 2.9 and 2.10]{Chacholski2020}). 
In other words, the remaining, if any, $p_0\coloneqq n_0-\beta_0$ $0$-simplices are paired with as many $1$-simplices. 
If $p_1\coloneqq n_1-\beta_1-p_0>0$, then we have $p_1$-many positive $1$-simplices paired with as many $2$-simplices, and so forth for all degrees up to the maximal dimension of $K$. 
Thus, $m'=\sum_{i=0}^{\dim(K)}\beta_i$ and $m=\dfrac{1}{2}\left(-m'+\sum_{i=0}^{\dim(K)}n_i\right)$.
\end{proof}

\subsection{Compatible functions}

As we are studying evolving functions, we next give an explicit definition for relating them based on persistence.

\begin{definition}
    Let $f,g\colon K\to \mathbb{R}$ be two filtration functions. 
The functions $f$ and $g$ are \define{compatible} if $f(\sigma) < f(\tau)$ $\implies$ $g(\sigma)\leq g(\tau)$ and $g(\sigma') < g(\tau')$ $\implies$ $f(\sigma')\leq f(\tau')$.
\end{definition}

\begin{examplet}
\label{ex:L_compatiblefunctions}
In order to understand \cref{def:comb_equiv_compatible_functions}, we next introduce an example which will be used throughout the paper, which is adapted from the excellent example of \cite[Fig.~4-6]{Hickok2026}. 
We will emphasize this particular example by calling the simplicial complex shown at left in  \cref{fig:CompatibleFunctions} by $L$. 
Throughout, we will focus on the 1-dimensional homology induced by the choice of function value on the simplices $a$, $b$, $c$, and $d$. 
We can assume that all other vertices and edges entered at some earlier time than these four (WLOG at value 0 for these examples), noting that they will be internally paired in persistence to give the 0-dimensional diagram: the infinite 0-dimensional class with one of the vertices, and three other (vertex, edge) pairs for the rest. 
For each function in \cref{fig:CompatibleFunctions}, we show the portion of the Hasse diagram for the compatible order corresponding to the simplices $\{a,b,c,d\}$. 
The functions $f$ and $g$ are compatible: for example, $g(a) < g(b)$ and $f(a) = f(b)$. 
Note that $f$ and $h$ are also compatible. 
However, $g$ and $h$ are not compatible since $g(a) <g(b)$ but $h(b)<h(a)$. 
We provide this example in part to emphasize the fact that compatibility is \textit{not} an equivalence relation. 
\end{examplet}

\begin{figure}[h]
    \centering
    \includegraphics[width=0.7\linewidth]{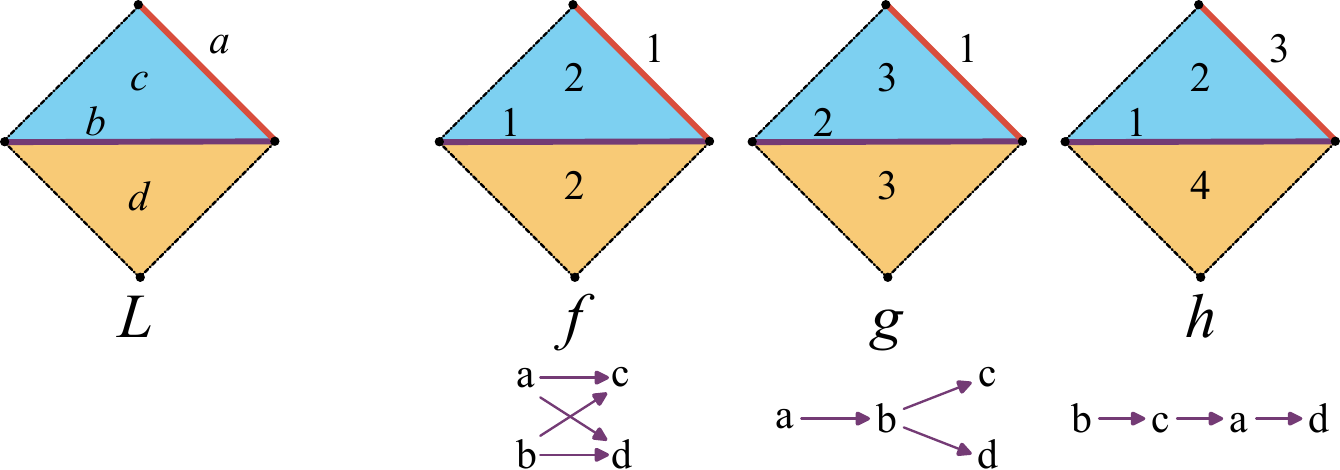}
    \caption{Example adapted from \cite{Hickok2026}. At left is the simplicial complex $L$ used as a running example. At right are three functions $f,g,h\colon K \to \R$ specified on edges $a$ and $b$, and triangles $c$ and $d$. Below each function is the Hasse diagram of the minimal compatible orderings $\prec_f$, $\prec_g$, $\prec_h$ for the named simplices. We assume that all other vertices and edges have function value 0 but generally ignore them for the sake of our examples. }
    \label{fig:CompatibleFunctions}
\end{figure}

We next give a definition to relate the pairs arising from two different consistent orderings for the same filtration function.
\begin{definition}
\label{def:filt_equiv}
    Let $f\colon K\to \mathbb{R}$ be a filtration function and $\prec_1$ and $\prec_2$ be two consistent orderings with respect to $f$. 
    A \define{filtration equivalence of pairings} $P_{\prec_1}$ and $P_{\prec_{2}}$ is an isomorphism $\nu \colon P_{\prec_1} \to P_{\prec_{2}}$ such that for all $\nu(\sigma,\tau) = (\sigma',\tau')$,
    \begin{itemize}
        \item $(f(\sigma),f(\tau)) = (f(\sigma'),f(\tau'))$, and 
        \item if $\sigma \neq \sigma'$ and $\tau \neq \tau'$, then there exists a consistent ordering $\prec_S$ of $f$ $(\sigma,\tau'),(\sigma',\tau)\in P_{\prec_s}$. 
    \end{itemize}
    If $\nu(\sigma,\tau) = (\sigma',\tau')$, we say $(\sigma, \tau) \in P_{\prec_1}$ and $(\sigma',\tau') \in P_{\prec_{2}}$ are \define{filtration equivalent with respect to $\nu$}. 
\end{definition}

The idea behind this definition is that there is some choice of consistent order and pairings with respect to $f$ such that both $\sigma$ and $\sigma'$ satisfy the split conditions (\hyperlink{SC1}{\bf SC1-2-3}) both with $\tau$ and $\tau'$, which does not happen if the simplices correspond to two different features (such as, two different components). The condition on $\nu$ tells us that $\sigma$ and $\tau$ split together, as well as $\sigma'$ and $\tau'$, and to ensure that they are representing ``the same feature'', we added the other possible pairings.

\begin{examplet}
\label{ex:combinatorialEquivalentSingleFunction}
Continuing the example of \cref{ex:L_compatiblefunctions},
 consider the function $g$ shown in \cref{fig:CompatibleFunctions}. 
We have two options for total orders refining the shown minimal compatible ordering: $a \prec_1 b \prec_1 c \prec_1 d$ and $a \prec_{2} b \prec_{2} d \prec_{2} c$. 
The first results in persistence pairing $P_{\prec_1} = \{(a,d), (b,c)\}$ and the second results in pairing $P_{\prec_{2}} = \{(a,c), (b,d) \}$. 
Because $f(a) \neq f(b)$, the only available  combinatorial equivalence is $\nu \colon (a,d) \mapsto (a,c);\, (b,c) \mapsto (b,d)$. 

Next consider the function $f$ in \cref{fig:CompatibleFunctions}. 
In this case, we take both $\prec_1$ and $\prec_2$ to be the same total order $b \prec a \prec d \prec c$, which results in pairings $P_{\prec} = \{(a,c), (b,d)\}$. 
There is an obvious isomorphism filtration equivalence given by the identity on $P_{\prec}$. 
However, if we choose $\prec_S$ to be $a \prec_S b \prec_S c \prec_S d$, which has pairings $P_{\prec_S} = \{(a,d), (b,c) \}$,
we have that the isomorphism  $\nu\colon P_\prec \to P_{\prec};\, (a,c) \mapsto (b,d)$ and $(b,d) \mapsto (a,c)$ is a filtration equivalence.
For example, fixing $(\sigma,\tau) = (a,c)$ and $(\sigma',\tau') = (b,d)$, we have that $(\sigma',\tau)$ and $(\sigma,\tau')$ are both in $P_{\prec_S}$. 
\end{examplet}

\begin{proposition}\label{prop:compatible_pairing}
Let $f$ and $g$ be compatible filtration functions. 
For any choice of consistent orderings and pairings $\prec_f$ with $P_{\prec_f}$, and $\prec_g$ with $P_{\prec_g}$, there is an order $\prec'$ and pairing $P_{\prec'}$, compatible with both $f$ and $g$, which gives filtration equivalences 
$\nu_1\colon P_{\prec_f} \to P_{\prec'}$ with respect to $f$ 
and 
$\nu_2\colon P_{\prec'} \to P_{\prec_g}$ with respect to $g$. 
\end{proposition}

\begin{proof}
By \cite[Lemma 4.4]{Chacholski2026}, given $f$ and $g$ compatible filtration functions, there exists a pairing $P_{\prec'}$ such that the sublevelset filtered chain complexes of $f$ and $g$ decompose into interval spheres given by the pairing $P_{\prec'}$. 
The isomorphisms $\nu_1$ and $\nu_2$ are then given by the isomorphism in defined in the proof of \cref{dgm_welldef_wrt_chosenpairing} applied twice: once between $P_{\prec_f}$ and $P_{\prec'}$ seen as consistent ordering with respect to $f$, and once between $P_{\prec'}$ and $P_{\prec_g}$ seen as consistent ordering with respect to $g$.
\end{proof}

The concept and term behind the next definition have been taken from the dissemination of the work \cite{Brooks2026} in research conversations with the authors.

\begin{definition}
\label{def:comb_equiv_compatible_functions}
    Let $f,g\colon K\to \mathbb{R}$ be two compatible filtration functions with total orderings $\prec_f$ and $\prec_g$.
    A \define{combinatorial equivalence of $P_{\prec_f}$ and $P_{\prec_g}$} is an isomorphism 
    $\nu = \nu_2 \circ \nu_1 \colon P_{\prec_f} \to P_{\prec_g}$  given by \cref{prop:compatible_pairing}.
    If such an equivalence exists, we say $P_{\prec_f}$ and $P_{\prec_g}$ are 
    \define{combinatorially equivalent}.
\end{definition}

\begin{examplet}
\label{ex:L_combinatorialEquivalent}
Continuing \cref{ex:L_compatiblefunctions,ex:combinatorialEquivalentSingleFunction}, we consider functions $f$ and $g$ from \cref{fig:CompatibleFunctions} which we have already noted are compatible. 
Say we are given total order refinements $b \prec_1 a \prec_1 c \prec_1 d$ for $f$ and $a \prec_2 b \prec_2 d \prec_2 c$ for $g$.
Then setting $\prec'$ to be $a \prec' b \prec' c \prec' d$, we have that $\prec'$ is a consistent ordering for both $f$ and $g$, resulting in the pairing $P_{\prec'} = \{ (a,d), (b,c) \}$. 
The only allowed filtration equivalence $\nu_2\colon P_{\prec'} \to P_{\prec_2}$ is $(a,d) \mapsto (a,c)$, $(b,c) \mapsto (b,d)$. 
However, there are two choices of isomorphism for $\nu_1 \colon P_{\prec_1} \to P_{\prec'}$, so we have either a filtration equivalence given by the composition
\begin{align*}
\begin{matrix}
   P_{\prec_1} & \xrightarrow{\nu_1} & P_{\prec'}& \xrightarrow{\nu_2} & P_{\prec_2}\\ 
   (b,d) & \mapsto & (a,d)  & \mapsto & (a,c) \\
   (a,c) & \mapsto & (b,c) & \mapsto & (b,d)
\end{matrix}
\end{align*}
or a filtration equivalence given by the composition
\begin{align*}
\begin{matrix}
   P_{\prec_1} & \xrightarrow{\nu_1'} & P_{\prec'}& \xrightarrow{\nu_2} & P_{\prec_2}\\ 
   (b,d) & \mapsto & (b,c) & \mapsto & (b,d)\\
   (a,c) & \mapsto & (a,d)  & \mapsto & (a,c). 
\end{matrix}
\end{align*}
\end{examplet}

\subsection{Distances}

As we will be studying maps $B \to \DgmAsp$ and $B \to \DgmDsp$, we need a notion of a metric on the diagrams to be able to define a topology and thus have a well-defined notion of continuity. 
For $\DgmDsp$, we choose the standard $q$-th Wasserstein distance, where we can fix any $q \in [1,\infty)$. 

\begin{definition}
Fix $q \in [1,\infty)$. 
Given a pair of D-diagrams $X$ and $Y$, considered as multisets of points in $\mathring{H}$, a partial matching is a bijection $\eta\colon X' \to Y'$ for subsets $X' \subseteq X$ and $Y' \subseteq Y$. 
The cost of a partial matching is given by 
\[
c(\eta) = \left( 
\sum_{x \in X'} \|x-\eta(x)\|_q^q 
+ \sum_{z \in (X \setminus X') \cup( Y \setminus Y')} \|z - \Delta\|_q^q
\right)^{1/p}.
\]
The \define{$q$-th Wasserstein distance} is 
\[
W_q(X,Y) = \inf_{\eta} c(\eta)
\]
where the infimum is taken over all partial matchings of $X$ and $Y$.
\end{definition}
    
\begin{definition}
    The space of persistence D-diagrams with the $q$-Wasserstein distance, $q \in \N$ is denoted $\DgmDsp$, or $\DgmDsp_q$ when $q$ needs to be emphasized. 
\end{definition}

Following \cite{Chacholski2026}, we have a related distance for the augmented persistence diagrams.
\begin{definition}
Fix $q \in [1,\infty)$. 
Given a pair of A-diagrams $X$ and $Y$ with the same cardinality, considered as multisets of points in $H$, a perfect matching is a bijection $\eta\colon X \to Y$. 
The cost of a partial matching is given by 
\[
c(\eta) = \left( \sum_{x \in X} \|x-\eta(x)\|_q^q \right)^{1/p}
\]
The \define{$q$-th Wasserstein distance} is 
\[
W_q(X,Y) = \inf_{\eta} c(\eta)
\]
where the infimum is taken over all perfect matchings of $X$ and $Y$.
\end{definition}

\begin{definition}
The space of persistence A-diagrams with the $q$-Wasserstein distance is denoted $\DgmAsp$, for $q\in\mathbb{N}$, or $\DgmAsp_q$ when $q$ needs to be emphasized.
\end{definition}

\subsection{Fibered filtration functions} 
\label{ssec:fiberedFiltFcn}

The main idea of this work is to study evolving functions on a simplicial complex. 
So, our assumed input is as follows.

\begin{definition}[\cite{Hickok2026}]
    A \define{fibered filtration function} is a  function 
    $g\colon K \times B \to \R$, where $B$ is a topological space, $K$ is a finite simplicial complex,  $g_p\colon K \to \R;\, g_p(\sigma) \coloneqq g(p,\sigma)$ is a filtration function on $K$ for every $p \in B$, and $g_\sigma\colon B \to \R$ is continuous for every $\sigma \in K$. 
\end{definition}

Note that this definition assumes a fixed simplicial complex $K$. 
Generalizations to a complex dependent on the point $p \in B$ would be an interesting extension of this work, but are outside the scope of what we do here. 
In any case, assuming a fixed $K$ encompasses many use-cases of our interest, including the persistent homology transform and its variants. 

Combining this definition with the A-diagram and d-diagrams from above, we have an induced map to persistence diagrams:

\begin{definition}
Given a fibered filtration function $g\colon K \times B \to \R$, the \define{induced A-persistence map} is 
\begin{equation*}
\begin{array}{rccc}
G_A\colon   & B & \longrightarrow &  \DgmAsp  \\
     & p & \mapsto & \DgmA(f_p)
\end{array}
\end{equation*}
Similarly, the \define{induced D-persistence map} is 
\begin{equation*}
\begin{array}{rccc}
G_d \colon  & B & \longrightarrow &  \DgmDsp  \\
     & p & \mapsto & \DgmD(f_p)
\end{array}
\end{equation*}
When obvious from the context, we will drop the $A$ and $D$ subscripts.

Whether or not they came from a fibered filtration function, we say a \define{$B$-parameterized (A-; respectively D-) diagram} is a function of the form $B\to \DgmAsp$ (respectively $B \to \DgmDsp$).
\end{definition}

\begin{theorem}\label{dgms_welldef_continuous}
Given a fibered filtration function, $g\colon K \times B \to \R$, both $G_A$ and $G_D$ are uniquely defined and are continuous in their respective metrics. 
\end{theorem}

\begin{proof}
The uniqueness is a direct result of \cref{dgm_welldef_wrt_consistentordering}.
Continuity of $G_D$ comes from the stability theorems such as \cite{Chazal2016,skrabaturner2020}.
Continuity of $G_A$ comes from \cite[Theorem 6.6]{Chacholski2026}.
\end{proof}

\begin{examplet}
\label{ex:L_2-param}
We conclude this section with the full, 2-parameter version of our motivating example described in \cref{ex:L_compatiblefunctions}, adapted from \cite{Hickok2026}. 
Given the simplicial complex $L$ from \cref{fig:CompatibleFunctions}, we define the fibered filtration function 
\begin{equation}
\label{eq:AbbyExampleFunction}
\begin{matrix}
f \colon & L \times \R^2 & \longrightarrow & \R \\
& (a,x,y) & \longmapsto & s(y)+2\\
& (b,x,y) & \longmapsto & -s(y) + 2\\
& (c,x,y) & \longmapsto & s(x) + 4\\
& (d,x,y) & \longmapsto & -s(x) + 4\\
& (\sigma, x,y) & \longmapsto & 0 & \text{ for all other } \sigma
\end{matrix}
\end{equation}
where $s(t) = \frac{1-e^{-t}}{1+e^{-t}}$ is the shifted sinusoidal function taking values in $[-1,1]$; see \cref{fig:AbbyExamplePosets} for a sketch. 
The exact function does not matter as much as the order induced in the various quadrants, which is visualized in the center of \cref{fig:AbbyExamplePosets}. 
The induced 1-dimensional map $G_A \colon B \to \DgmAsp$ gives a diagram $G_A(p)$ for each $(x,y) \in \R^2$ with exactly two off-diagonal points (in this case, nothing is on the diagonal in dimension 1, however, all finite 0-dimensional points will be on the diagonal). 
The function value of the two points evolves with the choice of point $(x,y) \in \R^2$. 
For this example, the diminished diagram function $G_D \colon \R^2 \to \DgmDsp$ results in the same diagrams since all points are off-diagonal. 
\end{examplet}

\begin{figure}[h]
\centering
\begin{minipage}{.2\textwidth}
\centering 

\includegraphics[width =.5\textwidth,align=c]{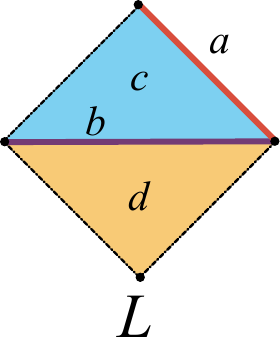}

\includegraphics[width = \textwidth, align = c]{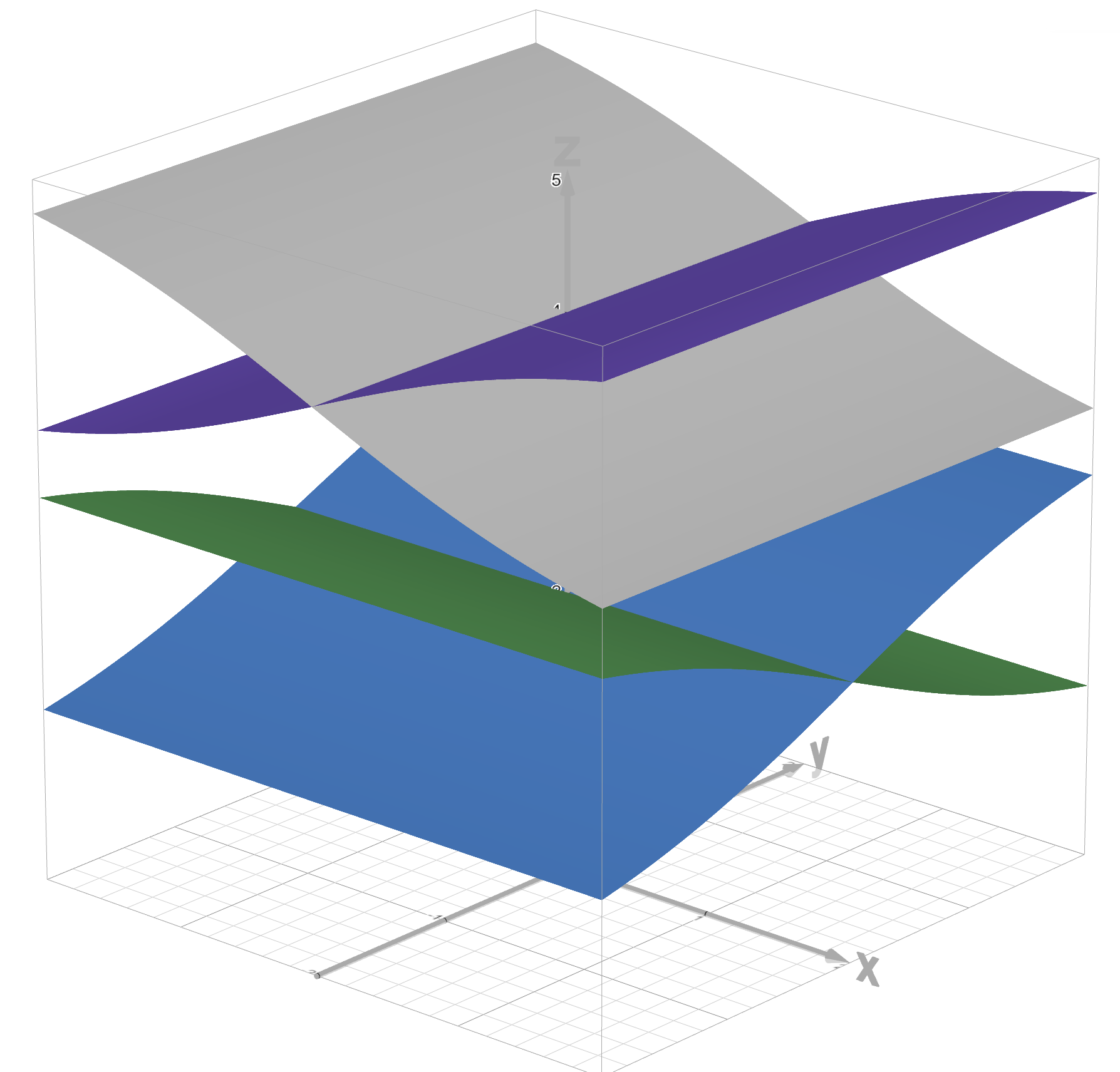}
\end{minipage}
\quad
\includegraphics[width=0.35\linewidth,align=c]{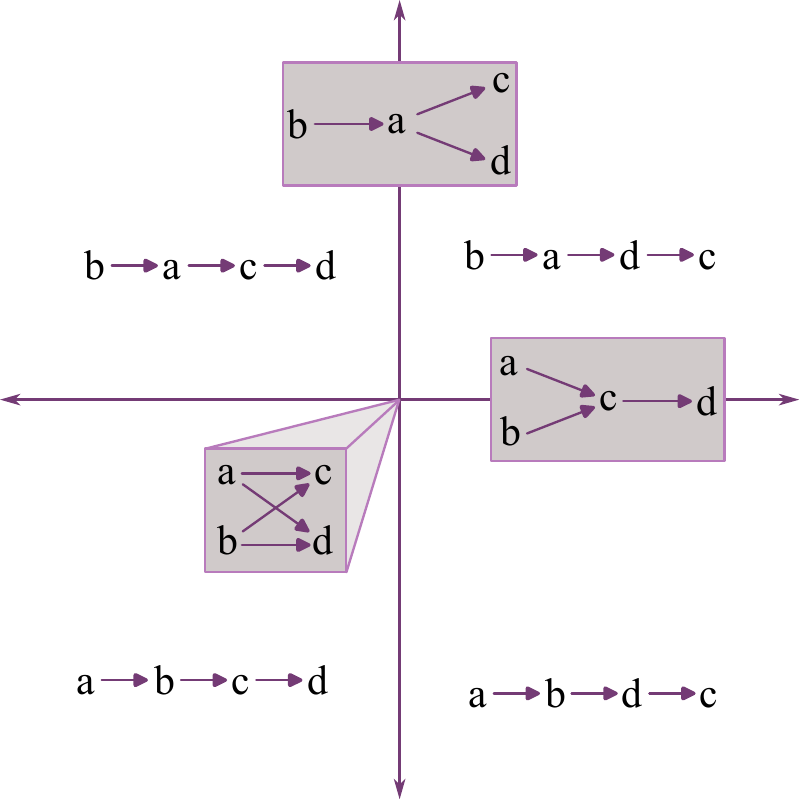}
\includegraphics[width=0.35\linewidth,align=c]{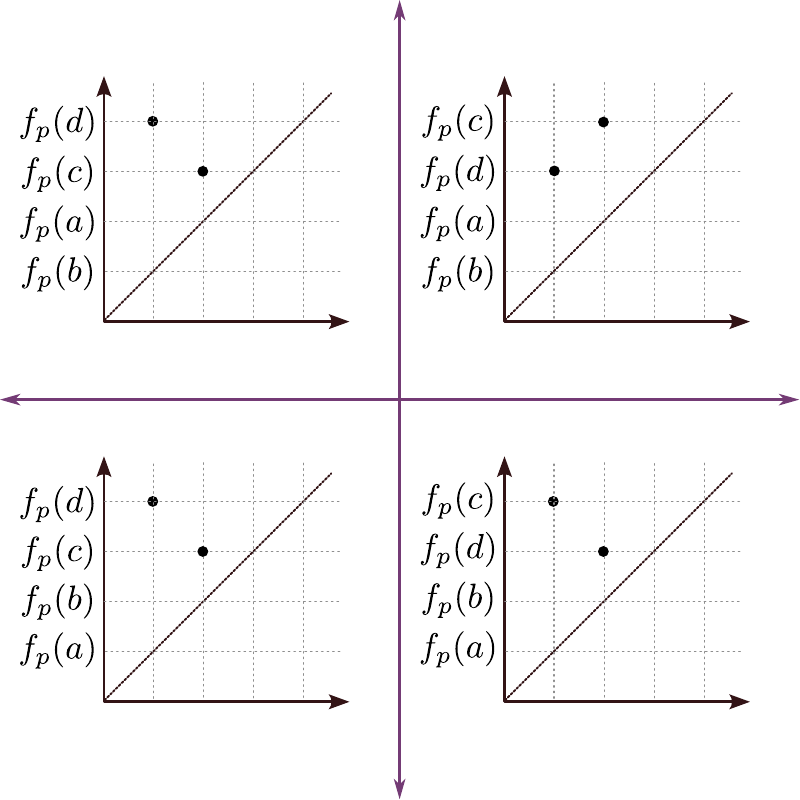}
    \caption{Left is the running example simplicial complex $L$, followed by a graph of the function in \cref{eq:AbbyExampleFunction} with $f_a$ in blue, $f_b$ in green, $f_c$ in purple, and $f_d$ in grey; see \href{https://www.desmos.com/3d/iqdanaupqx}{desmos.com/3d/iqdanaupqx} for an interactive plot.
    Next is the compatible order for each quadrant, and finally the resulting persistence diagrams in each region.
    }
    \label{fig:AbbyExamplePosets}
\end{figure}

\section{A-Canopies}
\label{sec:A-Canopies}

Now we assume we have fixed a fibered filtration function $g\colon K \times B \to \R$
with induced $B$-parameterized diagrams functions $G_A\colon B \to \DgmAsp$ and $G_D\colon B \to \DgmDsp$.
We will develop the structure of \emph{canopies} to encode these induced functions first without assuming that a fixed fibered filtration was given.
Then in subsequent sections we will show when a fibered filtration function has such canopies. 

\subsection{Definition and construction of an A-canopy}
Recall that we use the notation $[N] =\{1,\ldots,N\}$ to be the set of $N$ elements with the discrete topology. 

\begin{definition}
\label{def:A-canopy}
An \define{augmented canopy} (or simply an \define{A-canopy}) is a collection $\cE = (E,B,\pi,N,  \eta)$ where $B$ is Hausdorff space, $(E,B,\pi)$ is a bundle where $|\pi\inv(p)| = N \in \Z _{\geq 0}$ for all $p \in B$, and $\eta\colon E \to H$ is continuous.
If, additionally, there is a set $S\subset B$ such that
$$ (E\setminus \pi\inv(S),B \setminus S,\pi,[N])$$
is a fiber bundle, we say that $\cE$ is an \define{approximately fibered A-canopy}. 
If $S  = \emptyset$, (i.e.~$(E,B,\pi,[N])$ is a fiber bundle), we say that $\cE$ is a \define{fibered A-canopy}. 

\end{definition}

Consider the cartoon example of \cref{fig:A-canopy-example}. 
In this construction, we have a base space $B$ and $E$ which are both the circle $\S^1$, with $\pi$ acting as a degree 2 cover map. 
The preimages of three points in $B$ are shown as the purple triangles, the red circles, and the blue squares.
The map $\eta$ provides the location in $H$ corresponding to each point in $E$. 
The intuition here is that the image in $\eta$ of the preimages of a point $p \in B$ will end up representing the augmented persistence diagram associated to $p$. 
So in this example, we have two points in the persistence diagram at any fixed $p$; in particular for the red circle point, there is one off-diagonal and one on-diagonal point in the augmented diagram. 
Also note that, while not shown in the figure, the image $\pi(E)$ can be a non-embedding, even a non-immersion. 
For this reason, we are careful to treat the image under $\eta$ as a set with multiplicity since our diagram, which will be represented by $\eta(\pi\inv(p))$, is a set of points in $H$ with multiplicity.
\medskip

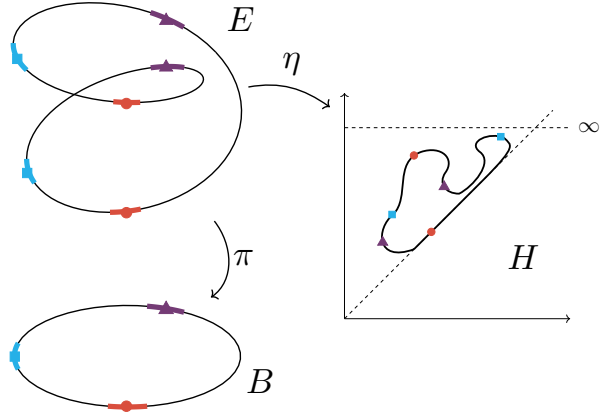
\begin{figure}[h]
\centering
\resizebox{.5\linewidth}{!}{
\begin{tikzpicture}
\def\firstball{lizRed} 
\def\firstarc{lizRed} 
\def\secondball{lizDkPur} 
\def\secondarc{lizDkPur} 
\def\thirdball{lizBlue} 
\def\thirdarc{lizBlue} 
\node[scale=2] (doubleLoopLabel) at (2.3,6.8) {$E$};
\node (doubleLoop) at (0,5) {
\begin{tikzpicture}
\begin{axis}[
    hide axis,
    view={-40}{45}, 
    width=8cm,
    height=8cm,
]
\def\R{5} 
\def\w{1.0} 
\addplot3[
    samples y=0,
    thick,
    samples=60,
    domain=0:360,
    variable=x,
    black, 
] (
    {(\R + \w*cos(x/2)) * cos(x)},
    {(\R + \w*cos(x/2)) * sin(x)},
    {\w*sin(x/2)}
);
\addplot3[
    samples y=0,
    thick,
    samples=70,
    domain=0:360,
    variable=x,
    black, 
] (
    {(\R - \w*cos(x/2)) * cos(x)},
    {(\R - \w*cos(x/2)) * sin(x)},
    {-\w*sin(x/2)}
);

\addplot3[
    samples y=0,
    line width=3pt,
    samples=70,
    domain=225:241,
    variable=x,
    \firstarc, 
] (
    {(\R - \w*cos(x/2)) * cos(x)},
    {(\R - \w*cos(x/2)) * sin(x)},
    {-\w*sin(x/2)}
);
\addplot3[
    samples y=0,
    line width=3pt,
    samples=60,
    domain=226:243,
    variable=x,
    \firstarc, 
] (
    {(\R + \w*cos(x/2)) * cos(x)},
    {(\R + \w*cos(x/2)) * sin(x)},
    {\w*sin(x/2)}
);

\addplot3[
    samples y=0,
    line width=3pt,
    samples=70,
    domain=0:30,
    variable=x,
    \secondarc, 
] (
    {(\R - \w*cos(x/2)) * cos(x)},
    {(\R - \w*cos(x/2)) * sin(x)},
    {-\w*sin(x/2)}
);
\addplot3[
    samples y=0,
    line width=3pt,
    samples=60,
    domain=20:37,
    variable=x,
    \secondarc, 
] (
    {(\R + \w*cos(x/2)) * cos(x)},
    {(\R + \w*cos(x/2)) * sin(x)},
    {\w*sin(x/2)}
);

\addplot3[
    samples y=0,
    line width=3pt,
    samples=70,
    domain=147:170,
    variable=x,
    \thirdarc, 
] (
    {(\R - \w*cos(x/2)) * cos(x)},
    {(\R - \w*cos(x/2)) * sin(x)},
    {-\w*sin(x/2)}
);
\addplot3[
    samples y=0,
    line width=3pt,
    samples=60,
    domain=135:165,
    variable=x,
    \thirdarc, 
] (
    {(\R + \w*cos(x/2)) * cos(x)},
    {(\R + \w*cos(x/2)) * sin(x)},
    {\w*sin(x/2)}
);
\end{axis}
\end{tikzpicture}
};

\node[scale=2] (baseSpaceLabel) at (2.7,-.5) {$B$};
\node (baseSpace) at (0,0) {
\begin{tikzpicture}
\begin{axis}[
    hide axis,
    view={40}{40},
    width=8cm,
]
\addplot3 [
    samples=50,
    domain=0:360,
    samples y=0,
    thick
] (
    {cos(x)},
    {sin(x)},
    {0});

\addplot3 [
    samples=50,
    domain=300:320,
    samples y=0,
    line width=3pt,
    \firstarc
] (
    {cos(x)},
    {sin(x)},
    {0});
\addplot3 [
    samples=50,
    domain=100:120,
    samples y=0,
    line width=3pt,
    \secondarc
] (
    {cos(x)},
    {sin(x)},
    {0});
\addplot3 [
    samples=50,
    domain=205:235,
    samples y=0,
    line width=3pt,
    \thirdarc
] (
    {cos(x)},
    {sin(x)},
    {0});

\end{axis}
\end{tikzpicture}
};
\node[scale=2] (DgmLabel) at (8,2) {$H$};
\node (Dgm) at (7,3) {
\adjustbox{scale=.7}{
\begin{tikzpicture}
\tikzset{
dot/.style = {circle, fill, minimum size=#1,inner sep=0pt, outer sep=0pt},
dot/.default = 6pt 
}
\def\xmax{6}
\def\ymax{6}
\draw[->, thick=2] (0,0) -- (\xmax + 0.5, 0);
\draw[->, thick=2] (0,0) -- (0, \ymax + 0.5);
\draw[dashed] (0,5.5) -- (\xmax + 0.5, 5.5);
\node[right,scale=1.8] at (\xmax + 0.5,5.5) {$\infty$};
\draw[thick=2,dashed] (0,0) -- (\xmax, \ymax);
 \draw[black,line width=1.5pt]
    (2,2)
    .. controls (0.9,1.6) and (0.9,2.6) .. (1.4,3.0)
    .. controls (1.9,3.5) and (1.4,4.5) .. (2.2,4.8)
    .. controls (2.8,5.0) and (3.2,4.6) .. (2.9,4.2)
    .. controls (2.7,3.9) and (2.9,3.6) .. (3.3,3.6)
    .. controls (3.7,3.8) and (4.3,4.5) .. (3.9,4.8)
    .. controls (3.5,5.1) and (4.2,5.4) .. (4.6,5.2)
    .. controls (4.8,5.0) and (5.0,5.0) .. (4,4)
    .. controls (4,4) and (2.0,2.0) .. (2.0,2.0);
    
\draw[fill=\firstball,draw=\firstball] (2.5,2.5) circle[radius=3pt];
\node[fill=\secondball,draw=\secondball,regular polygon, regular polygon sides=3, scale=0.5] at  (1.1,2.2) {};
\node[fill=\thirdball,draw=\thirdball,aspect=1, scale=0.8] at (1.37,3) {};
\draw[fill=\firstball,draw=\firstball] (2,4.7) circle[radius=3pt];
\node[fill=\secondball,draw=\secondball,regular polygon, regular polygon sides=3, scale=0.5] at  (2.87,3.8) {};
\node[fill=\thirdball,draw=\thirdball,aspect=1, scale=0.8] at (4.5,5.24) {};
\end{tikzpicture}
}
};

\draw[->, thick=1.5] (doubleLoop) to [out=-52,in=35] (baseSpace);
\node[right,scale=2] (gammaBLabel) at (1.9,2) {$\pi$};
\draw[->, thick=1.5] (doubleLoop) to [out=10,in=145] (Dgm);
\node[right,scale=2] (gammaHLabel) at (2.9,5.9) {$\eta$};

\draw[fill=\firstball,draw=\firstball] (0,-1) circle[radius=3pt];
\node[fill=\secondball,draw=\secondball,regular polygon, regular polygon sides=3, scale=0.5] at (0.8,0.95) {};
\node[fill=\thirdball,draw=\thirdball,aspect=1, scale=0.8] at (-2.25,0) {};
\draw[fill=\firstball,draw=\firstball] (0,2.9) circle[radius=3pt];
\draw[fill=\firstball,draw=\firstball] (0,5.1) circle[radius=3pt];
\node[fill=\secondball,draw=\secondball,regular polygon, regular polygon sides=3, scale=0.45] at  (0.8,5.83) {};
\node[fill=\secondball,draw=\secondball,regular polygon, regular polygon sides=3, scale=0.45] at  (0.8,6.77) {};
\node[fill=\thirdball,draw=\thirdball,aspect=1, scale=0.8] at (-2,3.7) {};
\node[fill=\thirdball,draw=\thirdball,aspect=1, scale=0.8] at (-2.2,6) {};
\end{tikzpicture}
}
\caption{An cartoon example of an A-canopy (see \cref{def:A-canopy}).
}
\label{fig:A-canopy-example}
\end{figure}

Assume we are given a fibered filtration function $g\colon K \times B \to \R$ and its induced map $G_A\colon B \to \DgmA$. 
The goal is to show that for nice enough inputs, there is an  A-bundle that encodes the same structure as $G_A$. 
Formally, this encoding is given in the following definition.
\begin{definition}\label{def:A-induced}
A $B$-parameterized diagram $G\colon B \to \DgmAsp$ \define{has an induced A-canopy} $\cE = (E,B,\pi,N, \eta)$ if $G(p) = \eta(\pi\inv(p))$ $\forall p \in B$, considered as a set with multiplicity. 
The \define{induced $B$-parameterized diagram} of an $A$-canopy $\cE = (E,B,\pi,N,\eta)$ is 
\[ 
\begin{matrix}
G_A^\cE \colon & B & \longrightarrow & \DgmAsp\\ 
& p & \longmapsto & \eta(\pi\inv(p)).
\end{matrix}
\]
where the image is a set with multiplicity.
\end{definition}

\begin{remark}
\label{rem:MetaDiagram-A-canopies}
In a sense, we are working in the following diagram:
\[
\begin{tikzcd}[column sep=0.5in]
\text{Fibered Filt.~Functions }K \times B \to \R 
\ar[r,"\text{induced}"]
& \Func(B,\DgmAsp) 
\ar[r,"\text{Thm.~\ref{thm:A-bundle-decomp}}", bend left]
& \text{A-canopies}
\ar[l,"\text{induced}", bend left]
\end{tikzcd}
\]
Obtaining the induced $B$-parametrized diagram from an A-canopy is straightforward from the set-up, hence the use of the term \textit{the} induced $B$-parameterized diagram. 
Moreover, it also follows immediately that if one has an induced A-canopy for a given $B$-parameterized diagram, the two-arrow loop in the diagram commutes in either order.
However, the existence of an induced A-canopy from a $B$-parametrized diagram is not immediate, and we show its existence in \cref{thm:A-bundle-decomp} so long as the input came from a fibered filtration function.
\end{remark}

Our task now is to build an A-canopy $\cE$ for a given $g\colon K \times B \to \R$. 
In particular, we assume we have fixed a collection of orderings $\{\prec_p \mid p \in B\}$ consistent for each $g_p\colon K \to \R$  with chosen pairings $\{P_{\prec_p} \mid p \in B\}$.
Then we define the set
\begin{equation*}
    E = \{ \left(p,\sigma, \tau\right) \mid p \in B, (\sigma,\tau) \in P_{\prec_p} \}. 
\end{equation*} 
Note that while this is similar to the construction of a persistence bundle in \cite{Hickok2026}, we encode the second and third coordinates of the points by pairs of simplices rather than points in $H$. 
Our first job is to build a topology on $E$. 

Since $\DgmAsp$ is a metric space with an induced topology, we have the notion of a path in diagram space, namely a continuous map $\gamma\colon [0,1] \to \DgmAsp$. 
However, we will need a stronger notion of paths to be able to track the pairings in the larger space $E$. 
We start by using the combinatorial equivalence for compatible functions (\cref{def:comb_equiv_compatible_functions}) to discretize an isomorphism between pairings arising from non-compatible functions.

\begin{definition}
\label{def:discretizeIsomorphism}
Fix two filtration functions $f,g\colon K \to \R$ with  orders $\prec_f$ and $\prec_g$ consistent with $f$ and $g$ respectively. 
Further, let $\nu\colon P_{\prec_f} \to P_{\prec_g}$ be given, along with a collection of filtration functions $\{f_i\colon K\to \mathbb{R}\}_{i=0}^{n}$ so that $f = f_0$, $g = f_n$, and each pair $f_i$ and $f_{i+1}$ are compatible.
If there exists a set of  consistent orders $\{\prec_i \coloneqq \prec_{f_i}\}_{i=0}^n$ with chosen pairs $\{P_{\prec_i}\}_{i=0}^n$ and combinatorial equivalences $\nu_i\colon P_{\prec_i} \to P_{\prec_{i+1}}$ such that $\nu = \nu_{n-1}\circ \cdots \circ \nu_1 \circ \nu_0$, then we say that \define{$\{f_i\}_{i=0}^n$ discretizes $\nu$}.
\end{definition}

\begin{figure}[h]
    \centering
\includegraphics[width=0.5\linewidth]{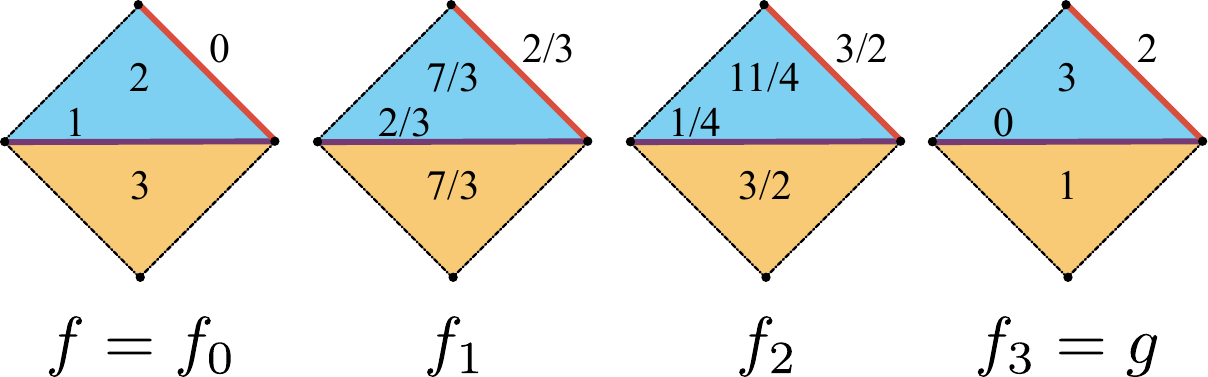}
    \caption{
    A sequence of functions on a simplicial complex $L$ from \cref{fig:CompatibleFunctions} used to elucidate the definition of a discretization of an isomorphism on pairs. 
    See details in \cref{ex:compatible_functions}
    }
    \label{fig:CompatibleFunctions-Complexes}
\end{figure}

\begin{examplet}
\label{ex:compatible_functions}
Consider the example $f$ (leftmost) and $g$ (rightmost) functions shown in \cref{fig:CompatibleFunctions-Complexes}, noting that they are not compatible with each other.  
The minimal consistent order for $f$ is $\prec_f \colon a \to b \to c \to d$, and the minimal consistent order for $g$ is $\prec_g \colon b \to d \to a \to c$. 
Then we choose $f_1$ and $f_2$ functions as shown, with compatible orders/pairings
\begin{align*}
    \prec_1 \colon &
    b \to a \to c \to d;\quad P_{\prec_1} = \{(a,c), (b,d) \} \\ 
    \prec_2 \colon & b \to d \to a \to c;\quad P_{\prec_2} = \{(b,d), (a,c) \}
\end{align*}
and sequence of combinatorial equivalences given by 
\begin{equation*}
  \begin{matrix}
   P_{\prec_f} & \xrightarrow{\nu_0} 
   & P_{\prec_1}& \xrightarrow{\nu_1} 
   & P_{\prec_2}& \xrightarrow{\nu_2} 
   & P_{\prec_g} \\ 
   (a,d) & \mapsto & (b,d)  & \mapsto & (b,d)& \mapsto & (b,d) \\
   (b,c) & \mapsto & (a,c) & \mapsto & (a,c)& \mapsto & (a,c)
\end{matrix}  
\end{equation*}
So this collection discretizes the isomorphism
\begin{equation*}
  \begin{matrix}
   P_{\prec_f} & \xrightarrow{\nu} 
   & P_{\prec_g} \\ 
   (a,d) &  \mapsto  & (b,d) \\
   (b,c) & \mapsto  & (a,c).
\end{matrix}  
\end{equation*}

As seen in previous examples, this is not the only isomorphism $P_{\prec_f} \to P_{\prec_g}$ that can be discretized. 
Choosing the other available combinatorial equivalence $\nu_0' \colon P_{\prec_f} \to P_{\prec_1}$ gives a sequence of combinatorial equivalences 
\begin{equation*}
  \begin{matrix}
   P_{\prec_f} & \xrightarrow{\nu_0'} 
   & P_{\prec_1}& \xrightarrow{\nu_1} 
   & P_{\prec_2}& \xrightarrow{\nu_2} 
   & P_{\prec_g} \\ 
   (a,d) & \mapsto & (a,c) & \mapsto & (a,c)& \mapsto & (a,c) \\
   (b,c) & \mapsto & (b,d)  & \mapsto & (b,d)& \mapsto & (b,d)
\end{matrix}  
\end{equation*}
which discretizes the isomorphism 
\begin{equation*}
  \begin{matrix}
   P_{\prec_f} & \xrightarrow{\nu'} 
   & P_{\prec_g} \\ 
   (a,d) &  \mapsto  & (a,c) \\
   (b,c) & \mapsto  & (b,d).
\end{matrix}  
\end{equation*}
\end{examplet}

Our goal is next to define a trajectory, which will be a path that follows the points in $E$ along some path in the base space, tracking relationships using the discretization of \cref{def:discretizeIsomorphism}. 

\begin{definition}
\label{def:trajectory}
Let $(p,\sigma,\tau)$ and $(p',\sigma',\tau')$ be two points in $E$. 
A \define{trajectory} 
$\Gamma \colon (p,\sigma,\tau) \rightsquigarrow  (p',\sigma',\tau') $ 
between them is a map $\Gamma \colon [0,1]\to E$ such that 
\begin{itemize}
\item $\Gamma_B \colon [0,1] \to B$, the projection onto the first factor of $E$, is continuous. 
\item $\Gamma(0) = (p,\sigma,\tau)$ and $\Gamma(1) = (p',\sigma',\tau')$. 
\item There is an isomorphism $\nu\colon P_{\prec_{f_p}} \to P_{\prec_{f_{p'}}}$ with $\nu(\sigma,\tau) = (\sigma',\tau')$ such that the following is true. For any $\{0 = t_0 <t_1<\cdots<t_k = 1\}$, there is a refinement $\{0 = s_0 <s_1<\cdots<s_n = 1 \} \supset \{t_i\}$ such that the functions 
$    \{f_{\Gamma_B(t)} \mid t \in \{s_i\} \} $
discretize $\nu$ in the sense of Defn.~\ref{def:discretizeIsomorphism}. 
\end{itemize}
We say a trajectory $\Gamma$ \define{projects to $U \subseteq B$} if $\Gamma_B([0,1]) \subseteq U$. 
\end{definition}

\begin{figure}[h]
\centering
\begin{minipage}{.15\textwidth}
    \includegraphics[width = \textwidth]{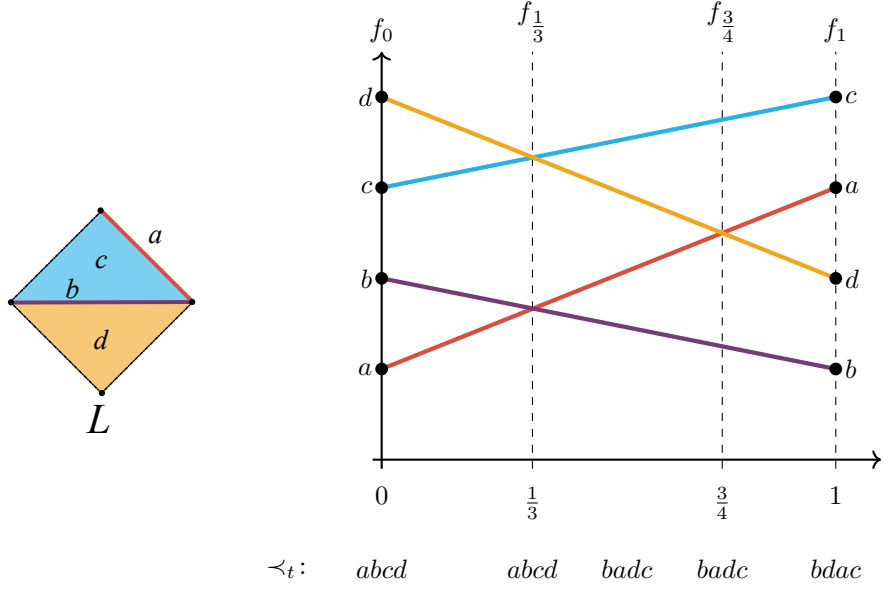}
\end{minipage}
\qquad
\begin{minipage}{.5\textwidth}
\begin{tikzpicture}[scale=1.2]

\def\xleft{0}
\def\xfa{1.66}  
\def\xfb{3.75} 
\def\xright{5}

\draw[thick, ->] (\xleft,-1.1) -- (\xleft,3.5);
\draw[dashed] (\xfa,-1.1) -- (\xfa,3.5);
\draw[dashed] (\xfb,-1.1) -- (\xfb,3.5);
\draw[dashed] (\xright,-1.1) -- (\xright,3.5);

\draw[thick, ->] (-.1,-1) -- (5.5, -1);


\draw[lizRed, ultra thick] (\xleft,0) -- (\xright,2); 
\draw[lizDkPur, ultra thick] (\xleft,1) -- (\xright,0); 
\draw[lizBlue, ultra thick] (\xleft,2) -- (\xright,3); 
\draw[lizYellow, ultra thick] (\xleft,3) -- (\xright,1); 

\fill (\xleft,0) circle (2pt) node[left] {$a$};
\fill (\xleft,1) circle (2pt) node[left] {$b$};
\fill (\xleft,2) circle (2pt) node[left] {$c$};
\fill (\xleft,3) circle (2pt) node[left] {$d$};

\fill (\xright,0) circle (2pt) node[right] {$b$};
\fill (\xright,1) circle (2pt) node[right] {$d$};
\fill (\xright,2) circle (2pt) node[right] {$a$};
\fill (\xright,3) circle (2pt) node[right] {$c$};

\node[below] at (\xleft,-1.2) {$0$};
\node[below] at (\xfa,-1.2) {$\tfrac{1}{3}$};
\node[below] at (\xfb,-1.2) {$\tfrac{3}{4}$};
\node[below] at (\xright,-1.2) {$1$};

\node[above] at (\xleft,3.5) {$f_0$};
\node[above] at (\xfa,3.5) {$f_{\tfrac{1}{3}}$};
\node[above] at (\xfb,3.5) {$f_{\tfrac{3}{4}}$};
\node[above] at (\xright,3.5) {$f_1$};

\node[below] at (-1,-2) {$\prec_t\colon$};
\node[below] at (\xleft,-2) {$abcd$};
\node[below] at (\xfa,-2) {$abcd$};
\node[below] at (2.7,-2) {$badc$};
\node[below] at (\xfb,-2) {$badc$};
\node[below] at (\xright,-2) {$bdac$};

\end{tikzpicture}
\end{minipage}

\caption{
A fibered filtration function for the simplicial complex $L$ shown for the simplices $a$, $b$, $c$, $d$. 
}
\label{fig:fiberedFiltFunc}
\end{figure}
\begin{examplet}
\label{ex:trajectory}
We show the definition of a trajectory on a 1-parameter fibered filtration function which is closely related to a linearized version of a 1-parameter slice of \cref{ex:L_2-param} along the line $y=x$. 
Consider the fibered filtration function $f\colon  K \times [0,1] \to \R$ where we show the functions\footnote{Note that $f_0, f_{1/3}, f_{3/4}$ and $f_1$ in \cref{fig:fiberedFiltFunc} are $f=f_0$  $f_1$, $f_2$, and $g=f_3$ respectively in \cref{fig:CompatibleFunctions-Complexes}.} 
$f_\sigma$ for $\sigma \in \{a,b,c,d \}$ in \cref{fig:fiberedFiltFunc}.
As in \cref{ex:L_compatiblefunctions} and other examples, we assume the other vertices and edges have some function value lower than $a$, $b$, $c$, and $d$; as they generate the 0-dimensional persistence, they can be safely ignored for the purposes of this example.

When we create the set $E$, we have assumed that there is a fixed compatible order for every point $p \in [0,1]$ with chosen pairings, so here we choose total compatible orders  
\begin{equation}
\label{eq:ExampleChosenOrders}
    \prec_t = 
    \begin{cases}
        a \to b \to c \to d & t \leq\tfrac{1}{3}\\ 
        b \to a \to d \to c & \tfrac{1}{3}<t\leq \tfrac{3}{4}\\
        b \to d \to a \to c & t > \tfrac{3}{4}
    \end{cases}
\end{equation}
resulting in pairings for the relevant simplices given by 
\begin{equation}
\label{eq:ExampleChosenPairings}
    P_{\prec_t} = 
    \begin{cases}
        \{ (a,d), (b,c) \}& t \leq\tfrac{1}{3}\\ 
       \{(b,d), (a,c) \} & \tfrac{1}{3}<t\leq \tfrac{3}{4}\\
       \{(b,d), (a,c) \}  & t > \tfrac{3}{4}.
    \end{cases}
\end{equation}
Then we can define a trajectory 
\begin{equation*}
    \Gamma (t) = 
    \begin{cases}
        (t, a,d) & t \in [0,\tfrac{1}{3}]\\
        (t, b,c) & t \in (\tfrac{1}{3},1]
    \end{cases}
\end{equation*}
since we can force any subdivison of $[0,1]$ to include values $\tfrac{1}{3}$ and $\tfrac{3}{4}$ and then use the discretization of $\nu$ arising from the composition
\begin{equation*}
  \begin{matrix}
   P_{\prec_0} & \xrightarrow{\nu_0} 
   & P_{\prec_{\tfrac{1}{3}}}& \xrightarrow{\nu_1} 
   & P_{\prec_{\tfrac{3}{4}}}& \xrightarrow{\nu_2} 
   & P_{\prec_1} \\ 
   (a,d) & \mapsto & (a,d)  & \mapsto & (b,d)& \mapsto & (b,d) \\
   (b,c) & \mapsto & (b,c) & \mapsto & (a,c)& \mapsto & (a,c)
\end{matrix}  
\end{equation*}
\end{examplet}

The first main result is that we can use these trajectories to define a topology on $E$. 

\begin{theorem}[A-canopies]
\label{thm:A-bundle-decomp} 
Fix a fibered filtration function $g\colon K \times B \to \cD$ with induced diagram function $G\colon B \to \DgmAsp$. 
Then $G$ has an induced A-canopy.
\end{theorem}

\begin{proof}
The proof is constructive to build $\cE = (E,B,\pi,N,  \eta)$. 
For every $f_p\colon K \to \R$ for $p \in B$, choose a compatible ordering $\prec_p$ for the function $f_p$. 
Then fix a chosen pairing $P_{\prec_p}$ for each $p \in B$. 
Let 
$$E = \{(p,\sigma, \tau) \mid p \in B, (\sigma,\tau) \in  P_{\prec_p} \}.$$ 
Note that $\pi\colon E \to B$ is easily defined as the projection onto the first factor. 
We define $\eta \colon E \to H$ by $\eta(p,\sigma,\tau) = (g_p(\sigma),g_p(\tau))$. 
By \cite[Paragraph 2.9, Prop. 3.2 and 3.3]{Chacholski2020}, this gives all of $G(p)$.
Further, since $\pi\inv(p) \cong P_{\prec_p}$, by Prop.~\ref{prop:Num_Pairs_Constant}, we know that $|\pi\inv(p)| =N$ is constant for some $N$. 
So, it only remains to show that $E$ is a topological space rather than just a set. 
We give $E$ a topology as follows. 
For a neighborhood $U \ni p$ in $B$ and a fixed $e = (p,\sigma,\tau) \in E$, we define the set 
\begin{equation*}
U_e = \{ e' = (p',\sigma',\tau') \mid \exists\,  \Gamma \colon e \rightsquigarrow e' \text{ that projects to }U \} \subseteq E.
\end{equation*}
Then we define $E$ to be a topological space with topology given by the basis 
\[
\mathcal{U} = \{U_e \mid e = (p,\sigma,\tau) \in E, U \ni p \text{ open in }B \} \, . 
\]
Lastly, by \cref{dgms_welldef_continuous}, $\eta$ is continuous.
\end{proof}

\begin{remark}\label{rem:pi_open}
A direct consequence of the way we defined the topology of $E$ in the proof of \cref{thm:A-bundle-decomp} is that the constructed map $\pi$ is open.
Indeed, if $V$ is open in $E$, then it is the union of basis elements $\bigcup_{i=1}^\infty U_e$ for $e \in E, U \ni \pi(e)$ open in $B$.
So $\pi(V) = \bigcup \pi(U_e) = \bigcup U$ by definition, which is an open set.
\end{remark}

We note here that a choice was made in the construction of $E$ in the proof of \cref{thm:A-bundle-decomp}, specifically for choices of compatible orderings and pairings for each $p$. 
However, in essence, this choice should not matter. 
Saying two A-canopies are the same is given by the following bundle-like definition of isomorphisms.

\begin{definition}
Two A-canopies $\cE_1 = (E_1,B,\pi_1,N,  \eta_1)$ and $\cE_2=(E_2,B,\pi_2,N,  \eta_2)$ are isomorphic if there is a homemorphism $\phi\colon E_1 \to E_2$ such that the diagrams
\begin{equation*}
\begin{tikzcd}
E_1 \ar[rr,"\phi"] \ar[dr, "\pi_1"'] && E_2 \ar[dl,"\pi_2"] 
& 
E_1 \ar[rr,"\phi"] \ar[dr, "\eta_1"'] && E_2 \ar[dl,"\eta_2"]\\
&B &&& H
\end{tikzcd}
\end{equation*}
commute.
\end{definition}

\begin{theorem}
\label{thm:A-bundle-unique}
The A-canopy of \cref{thm:A-bundle-decomp} is unique up to isomorphism.
\end{theorem}
\begin{proof}
We will show that for two choices of chosen pairings at each $p \in B$, there is an isomorphism given between the resulting canopies constructed in \cref{thm:A-bundle-decomp}. 
Say canopy $\cE_1 = (E_1,B,\pi_1,N,  \eta_1)$ was constructed using pairings $P_{p,1}$ for each $p \in B$; and canopy $\cE_2=(E_2,B,\pi_2,N,  \eta_2)$ was constructed using pairings $P_{p,2}$. 
For any fixed $p$ and $(\sigma,\tau) \in P_{p,1}$, \cref{prop:compatible_pairing} gives a combinatorial equivalence 
$\nu_p \colon P_{\prec_p,1} \to P_{\prec_p,2}$.
Define the map 
$\phi \colon E \to E'$ by 
$\phi(p,\sigma_1,\tau_1) = (p,\sigma_2,\tau_2)$ where 
$\nu_p(\sigma_1,\tau_1) = (\sigma_2,\tau_2)$. 
It is immediate by construction that $\pi_2 \phi = \pi_1$ and $\eta_2\phi = \eta_1$. 
So it remains to show that $\phi$ is a homeomorphism; in particular we must show it is a continuous bijection with a continuous inverse. 
The bijectivity is immediate from the use of \cref{prop:compatible_pairing}, so we check continuity of both $\phi$ and $\phi\inv$ by showing that it sends basis open sets to basis open sets. 
Specifically, we show $\phi(U_{e_1}) = U_{e_2}$ where $e_1 = (p,\sigma,\tau)$, $\phi(e_1) = e_2 = (p,\sigma',\tau') $, and $U \subset B$ is an open set. 

First, we need that $U_{e_2} \subseteq \phi(U_{e_1})$. 
Let $(q,\alpha_2,\beta_2) \in U_{e_2}$ be an arbitrary point. 
By definition of the basis elements, there is a trajectory 
$\Gamma\colon (p,\sigma_2,\tau_2) \rightsquigarrow (q,\alpha_2,\beta_2)$ given by a map $\Gamma \colon [0,1] \to E_2$ which projects to $U$. 
Let $(q,\alpha_1,\beta_1) = \phi \inv(q,\alpha_2,\beta_2)$. 
We claim that $\phi \inv\Gamma \colon [0,1] \to E'$ is a trajectory $(p,\sigma_1,\tau_1) \rightsquigarrow (q,\alpha_1,\tau_1)$. 
Indeed, note that the projection onto the first factor is unchanged, so it is continuous. 
Further 
$\phi \inv \Gamma(0) = \phi\inv(p,\sigma_2,\tau_2) = (p,\sigma_1,\tau_1)$ 
and likewise 
$\phi \inv \Gamma(1) = \phi\inv(q,\alpha,\beta_2) = (p,\alpha_1,\beta_1)$. 
Finally, by definition of a trajectory in $E_2$, there is an isomorphism $\nu \colon P_{p,2} \to P_{q,2}$ with the final property of \cref{def:trajectory}. 
To check the same property for $\phi \inv \Gamma$, let any 
$\{0 = t_0 <t_1<\cdots<t_k = 1\}$
be given, and let $\{0 = s_0 <s_1<\cdots<s_n = 1 \} \supset \{t_i\}$ be the given refinement such that the functions 
$  \{f_{\Gamma_B(t)} \mid t \in \{s_i\} \} $
discretize $\nu$. 
Then we pre- and post-compose the combinatorial equivalences $\nu_p$ and $\nu_q$ used in the beginning of the construction of $\phi$ to get a discretization of 
\begin{equation*}
\begin{tikzcd}
P_{p,1} \ar[r, "\nu_p"] \ar[rrr, bend right, "\tilde\nu"]
& P_{p,2} \ar[r,"\nu"] & P_{q,2} \ar[r, "\nu_q\inv"] & P_{q,1}
\end{tikzcd}
\end{equation*}
where $\tilde \nu$ is a fixed isomorphism since it is the composition of fixed isomorphisms, and so $\phi\inv\Gamma$ is a trajectory in $U_e$. 
In particular, this means we started with an arbitrary $(q,\alpha_2,\beta_2) \in U_{e_2}$, and showed $(q,\alpha_1,\beta_1) \in U_{e_1}$ with $\phi(q,\alpha_1,\beta_1)=(q,\alpha_2,\beta_2)$, completing the inclusion. 

Finally, we need to show that $\phi(U_{e_1}) \subseteq U_{e_2}$, which proceeds in a similar style to the other inclusion.
Let $(q,\alpha_2,\beta_2) \in \phi(U_{e_1})$ be given, and so there is a $(q,\alpha_1,\beta_1) \in E$ with $\phi(q,\alpha_1,\beta_1)=(q,\alpha_2,\beta_2)$. 
We need to show that $(q,\alpha_2,\beta_2)$ is in $U_{e_2}$. 
Since $(q,\alpha_2,\beta_2) \in \phi(U_{e_1})$, there is a trajectory $\Gamma \colon e_1 \rightsquigarrow (q,\alpha_1,\beta_1)$ given by a map $\Gamma\colon [0,1] \to E$.
We check that $\phi\Gamma \colon e_2 \rightsquigarrow (q,\alpha_2,\beta_2)$ is also a trajectory in $U_{e_1}$. 
As before, we have a fixed $\nu \colon P_{p,1} \to P_{q,1}$ given by the trajectory, so we use the isomorphism 
\begin{equation*}
\begin{tikzcd}
P_{p,2} \ar[r, "\nu_p\inv"] \ar[rrr, bend right, "\tilde\nu"]
& P_{p,1} \ar[r,"\nu"] & P_{q,1} \ar[r, "\nu_q"] & P_{q,2}
\end{tikzcd}
\end{equation*}
noting that any discretization of $\nu$ gives a discretization of $\tilde \nu$, completing the proof.
\end{proof}

\begin{examplet}
\label{ex:L_fullA-canopy}
We conclude the discussion of \cref{ex:trajectory}  by fully computing the A-canopy of the fibered filtration function shown in \cref{fig:fiberedFiltFunc}.
In particular, we will focus on the points in $\pi \inv(\tfrac{1}{3})$ and $\pi \inv (\tfrac{3}{4})$. 
For $p = \tfrac{1}{3}$, we have points $(\tfrac{1}{3},a,d)$ and $(\tfrac{1}{3},b,c)$ in $E$ from the pairings chosen earlier. 
Fix any small neighborhood $U \ni \tfrac{1}{3}$ in $B = [0,1]$. 
We can define a trajectory whose projection $\Gamma_B$ is constant by using the same function over and over in the discretization but changing the compatible order associated to it. 
Indeed, if we choose total orders 
\begin{equation*}
\prec_i =
\begin{cases}
a \to b \to c \to d & i = 0\\
b \to a \to c \to d & i = 1\\
b \to a \to d \to c & i = 2\\
a \to b \to d \to c & i = 3 \\
a \to b \to c \to d & i = 4
\end{cases}
\end{equation*}
which are all compatible with $f_{1/3}$, then the composition of combinatorial equivalences becomes 
\begin{equation*}
  \begin{matrix}
   P_{\prec_0} & \xrightarrow{\nu_0} 
   & P_{\prec_{1}}& \xrightarrow{\nu_1} 
   & P_{\prec_{2}}& \xrightarrow{\nu_2} 
   & P_{\prec_{3}}& \xrightarrow{\nu_3} 
   & P_{\prec_4} \\ 
   (a,d) & \mapsto & (b,d)  & \mapsto & (b,d)& \mapsto & (b,d) & \mapsto & (b,c)\\
   (b,c) & \mapsto & (a,c) & \mapsto & (a,c)& \mapsto & (a,c) & \mapsto & (a,d).
\end{matrix}  
\end{equation*}
This means that the point $(\tfrac{1}{3},a,d)$ is in the open set $U_{(1/3,b,c)}$ for any open set $U \ni \tfrac{1}{3}$; and similarly  $(\tfrac{1}{3},b,c)$ is in the open set $U_{(1/3,a,d)}$.
This means $(\tfrac{1}{3},a,d)$ and $(\tfrac{1}{3},b,c)$  are non-Hausdorff points of the constructed $E$. 

On the other hand, consider the points $(\tfrac{3}{4}, b,d)$ and $(\tfrac{3}{4}, a,c)$ in $E$. 
In this case, using the same orderings and pairings as previously (\cref{eq:ExampleChosenOrders,eq:ExampleChosenPairings}) and any small $U \ni \tfrac{3}{4}$, the only available combinatorial equivalences are
\begin{equation*}
  \begin{matrix}
   P_{\prec_{3/4-\e}} & \xrightarrow{\nu_0} 
   & P_{\prec_{3/4}}& \xrightarrow{\nu_1} 
   & P_{\prec_{3/4+\e}} \\ 
   (a,c) & \mapsto & (a,c)  & \mapsto & (a,c)& \\
   (b,d) & \mapsto & (b,d)  & \mapsto & (b,d)& 
   \end{matrix}  
\end{equation*}
This means that the open set $U_{3/4,a,c}$ cannot contain the point $(t,b,d) \in E$ for any $t \in U$. 
Similarly, the open set $U_{3/4,b,d}$ cannot contain the point $(t,a,c) \in E$ for any $t \in U$. 
Taken together, the resulting space $E$ with its topology is visualized in \cref{fig:ConstructedE}.
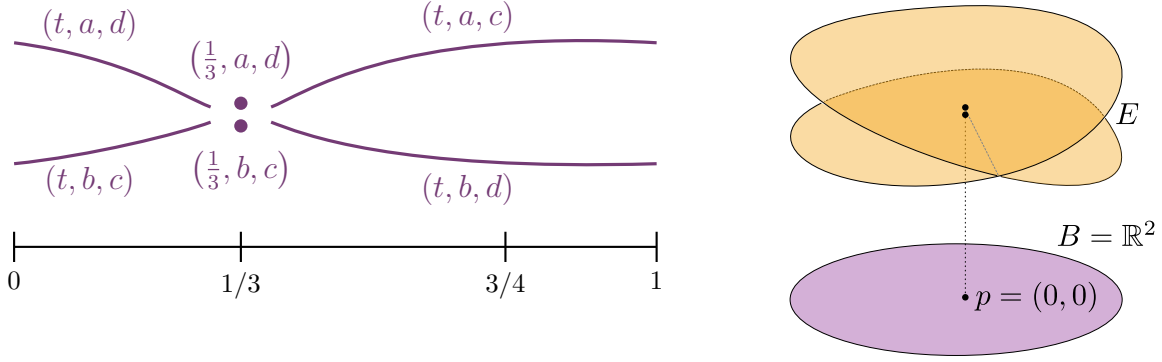
\begin{figure}[h]
    \centering
    \begin{minipage}{.5\linewidth}
    \begin{tikzpicture}

    \draw[thick] (0, -1.5) -- (8.5, -1.5);
    
    \draw[thick] (0, -1.3) -- (0, -1.7) node[below] {$0$};
    \draw[thick] (3, -1.3) -- (3, -1.7) node[below] {$1/3$};
    \draw[thick] (6.5, -1.3) -- (6.5, -1.7) node[below] {$3/4$};
    \draw[thick] (8.5, -1.3) -- (8.5, -1.7) node[below] {$1$};

    \draw[lizDkPur, very thick] (0, 1.2) .. controls (1.5, 1.0) and (2.2, 0.5) .. (2.6, 0.35);
    \draw[lizDkPur, very thick] (0, -0.4) .. controls (1, -0.3) and (2.2, 0.0) .. (2.6, 0.15);

    \fill[lizDkPur] (3, 0.4) circle (2.5pt);
    \fill[lizDkPur] (3, 0.1) circle (2.5pt);

    \draw[lizDkPur, very thick] (3.4, 0.35) .. controls (3.8, 0.5) and (5, 1.4) .. (8.5, 1.2);
    \draw[lizDkPur, very thick] (3.4, 0.15) .. controls (3.8, 0.0) and (5, -0.5) .. (8.5, -0.4);

    \node[lizDkPur, above] at (1, 1.1) {\large $(t, a, d)$};
    \node[lizDkPur, below] at (1, -0.3) {\large $(t, b, c)$};
    
    \node[lizDkPur, above] at (3, 0.6) {\large $\left(\tfrac{1}{3}, a, d\right)$};
    \node[lizDkPur, below] at (3, -0.1) {\large $\left(\tfrac{1}{3}, b, c\right)$};
    
    \node[lizDkPur, above] at (6, 1.2) {\large $(t, a, c)$};
    \node[lizDkPur, below] at (6, -0.4) {\large $(t, b, d)$};

\end{tikzpicture}
    \end{minipage}
    \qquad \qquad 
    \begin{minipage}{.35\linewidth}
    \includegraphics[width=\textwidth]{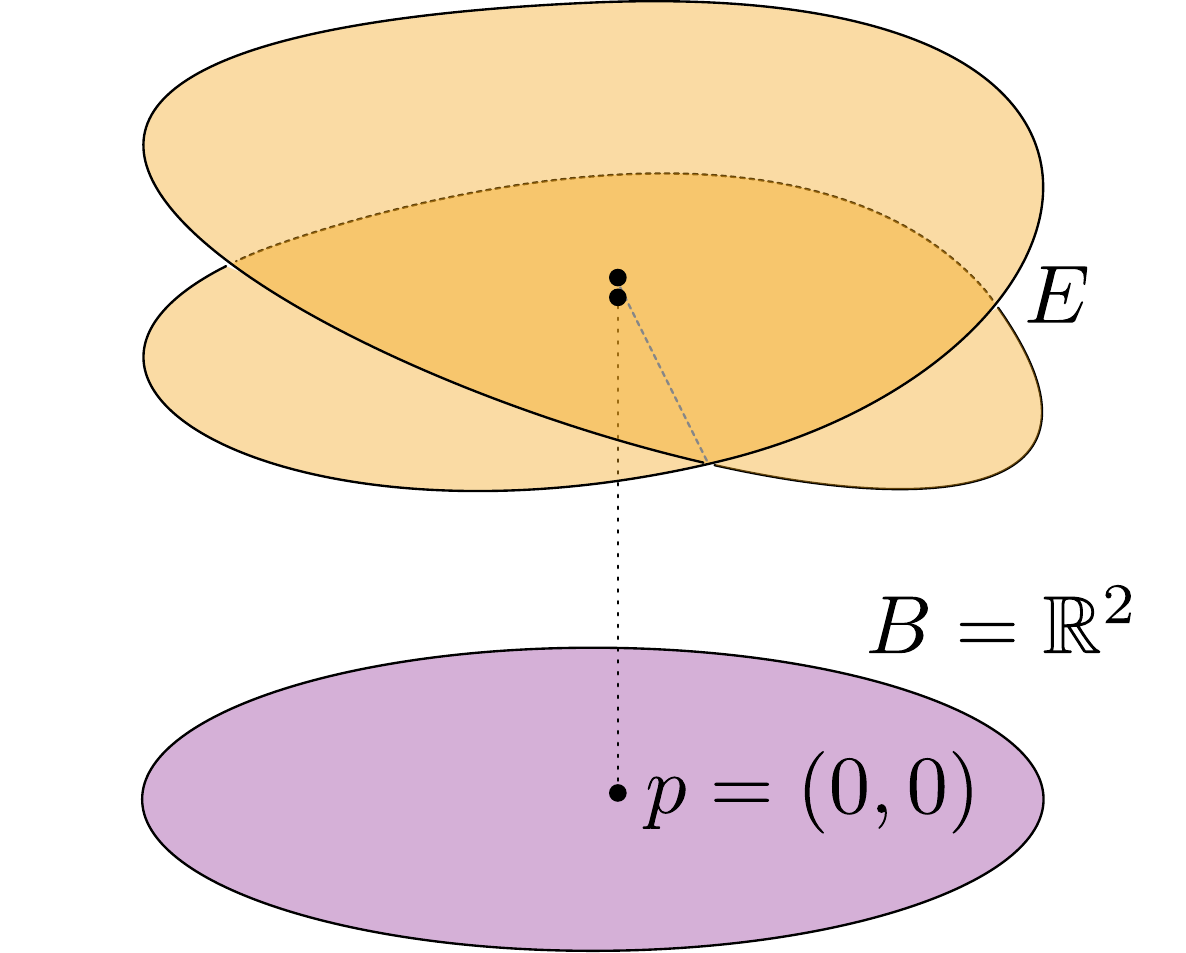}
    \end{minipage}
    \caption{Left: The space $E$ for the A-canopy of the fibered filtration function given in \cref{fig:fiberedFiltFunc} described in \cref{ex:L_fullA-canopy}. Note that there are two non-Hausdorff points over $t=1/3$. 
    Right: The space $E$ for the A-canopy of the 2-parameter fibered filtration of \cref{fig:AbbyExamplePosets} described in \cref{ex:L_2-param-fullCanopy}.
    Here, there is a pair of non-Hausdorff points over the origin; the grey dashed line in $E$ is meant to represent a non-intersection of the space.}
    \label{fig:ConstructedE}
\end{figure}
\end{examplet}

\begin{examplet}
\label{ex:L_2-param-fullCanopy}
We can also compute the A-canopy for the full 2D example from \cref{ex:L_2-param}. 
The behavior of the point at the origin is the same as the non-Hausdorff points in \cref{ex:L_fullA-canopy} over $t=1/3$ since that example was constructed from a 1-parameter slice of this example so we do not provide the details. 
However, the result of the 2-parameter version is that $E$ has a pair of non-Hausdorff points at the origin. 
The rest of the space has two points above every point in $\R^2$, and following a loop around the origin results in a degree 2 map down to the circle. 
\end{examplet}
\subsection{Structure of an A-canopy}

In \cref{thm:A-bundle-decomp}, we constructed an A-canopy from a fibered filtration function, but we did not discuss the second half of \cref{def:A-canopy}; namely which set $S$ results in the constructed A-canopy being approximately fibered?
In \cref{ex:largestS}, we show that finding a set $S$ for which an A-canopy is approximately fibered is trivial. 
However, the resulting set could be too large to be of any use. 
Therefore, in \cref{thm:A-bundle-decomp-generic}, we restrict our setting to prove that smaller sets $S$ can be found, providing the minimal one. 

\begin{examplet}(Largest $S$)\label{ex:largestS}
Let $g\colon K \times B \to \R$ be a fibered filtration function and let $\cE = (E,B,\pi,F)$ be an induced A-canopy obtained from  \cref{thm:A-bundle-decomp}.
Let 
$$S\coloneqq\{p\in B \mid  g(p,\sigma)=g(p,\tau) \text{ for some } \sigma,\tau \in K\}.$$
We note that, by continuity of $g(\cdot,\sigma)$ for all $\sigma \in K$, $S$ is closed as a subset of $B$. 
We will show that $\cE \setminus S = (E\setminus \pi\inv(S),B \setminus S,\pi,F)$ is an approximately fibered A-canopy. 
Indeed, if $B\setminus S=\emptyset$, the statement is trivially true. 
If $B\setminus S\neq \emptyset$, then $E \setminus \pi^{-1}(S)$ are the subsets of $E$ where all simplices have different filtration value. 
Moreover by the continuity of $g$, the ordering on the simplices induced by $g_p$ is total and fixed in each connected component of $\pi^{-1}(S)$, so the resulting pairings $P_{\prec_p}$ are equal through the component. 
The claim follows.
\end{examplet}

In the above example, all points in the base space for which at least two simplices have the same filtration value are excluded. 
This also excludes benign instances such as two simplices belonging to the same pair, i.e., points on the diagonal in the augmented persistence diagram, or pairs with the same birth time but different death times. 
Therefore, it is clearly too stringent. 
Thus, our next task is to show a minimal set $S$ for which there is an approximately fibered A-canopy.

We start by focusing on points in a topological space that break the Hausdorff assumption, \cref{def:non-Haus-point}. 
In particular, we can distinguish better or worse points of multiplicity in a persistence diagram $G_\cE(p)$ by whether there are non-Hausdorff points in the constructed $E$ as follows.

\begin{definition}
\label{def:Structural-Coincidental}
Fix a Hausdorff space $B$, an A-canopy $\cE = (E,B,\pi,N,  \eta)$ and the induced  $G_\cE\colon B \to \DgmAsp$, $p \mapsto \eta(\pi \inv(p))$ (\cref{def:A-induced}). 
A point $(b,d) \in G_\cE(p) \subset H$ with  $|\eta \inv(b,d)|\geq 2$ is called \define{coincidental} if all points in 
$\eta\inv(b,d) \cap \pi\inv(p)$ are Hausdorff in $E$.
Otherwise, $(b,d)$ is called \define{structural}. 
We abuse terminology to also call a point $(p,\sigma,\tau) \in E$ structural if it is a non-Hausdorff point of $E$ and $\eta(p,\sigma,\tau)$ is structural.
\end{definition}

\begin{examplet}
We have already seen examples of structural points, such as in \cref{ex:L_fullA-canopy}. 
To build a simple example with coincidental points, consider a simplicial complex with two connected components: the first has vertices $a_1$ and $b_1$ with edge $e_1$, the second has vertices $a_2,b_2$ with edge $e_2$. 
Say $f_p(a_i) = 0$, $f_p(b_i) = 1$, and $f_p(e_i) = 2$ for $i \in \{1,2\}$, then there will be a 0-dimensional point in the persistence diagram for the pairs $(b_1, e_1)$ and $(b_2,e_2)$ which are both at $(1,2)$ in $H$, but there is no combinatorial equivalence possible between them since they are in separate connected components.
\end{examplet}

Finally, we show that so long as we remove all non-Hausdorff points in $E$, the remaining A-canopy is fibered. 

\begin{theorem}[Fibered A-canopies]
\label{thm:A-bundle-decomp-generic}
Fix a fibered filtration function $g\colon K \times B \to \R$ for a Hausdorff space $B$, with induced diagram function $G\colon B \to \DgmAsp$ and A-canopy $\cE$ from \cref{thm:A-bundle-decomp}. 
Let 
\[S = \{p \in B \mid G(p) \text{ has structural points}\}.\] 
Then $\cE$ is approximately fibered with the set $S$.
Further, $\cE$ is not approximately fibered for any $S' \subsetneq S$, meaning this  $S$ is minimal. 
\end{theorem}
\begin{proof}
Denote the A-canopy by $\cE = (E,B,\pi,N,\eta)$. 
We need to show that
\[(E\setminus \pi\inv(S),B \setminus S,\pi,[N])\]
is a fiber bundle. 
This means that for every $p \in B \setminus S$, we need to find a neighborhood $p \in U \subset B \setminus S$ so that the diagram 
\begin{equation*}
\begin{tikzcd}
 \pi \inv(U) \ar[r, "\phi"] \ar[d,"\pi"'] & U \times [N] \ar[dl, "{\text{proj}_1}"] \\ 
U & 
\end{tikzcd}
\end{equation*}
commutes, where $\phi \colon \pi \inv(U) \to U \times F$ is a homeomorphism.

So, fixing $p \in B\setminus S$, we start by constructing $U$.
Since by \cref{prop:Num_Pairs_Constant}, $\pi\inv(p)$ has a constant number of points, we denote this by  $\pi\inv(p) = \{e_1,\ldots,e_{N} \}$.
By construction, these are all Hausdorff points, so there exists $V_i \ni e_i$ which is open in $E$ for $i \in [N]$ with $V_i \cap V_j = \emptyset$ for $i \neq j$. 
By \cref{rem:pi_open}, $\pi$ is an open map, so $\pi(V_i)$ is open in $B$ for all $i$. 
Then we define 
\[U =\bigcap_{i \in [N]} \pi(V_i)\]
which is an open set in $B$ containing $p$.
By construction, $\pi\inv(U) = \coprod V_i$ since the sets are pairwise disjoint. 
We define $\phi\colon \pi \inv(U) \to U \times [N]$ by $x \mapsto (\pi(x),i)$ if $x \in V_i$ and note that the $i$ is unique by assumption.

Next, we show that $\phi$ is a  homeomorphism. 
To see that it is injective, say $x$ and $y$ are in $\pi\inv(U)$ such that $\phi(x)=\phi(y) = (p',i)$. 
Because the second coordinate is the same, $x,y \in V_i$. 
Since the first coordinate is the same, $\pi(x) = \pi(y) = p' \in U$.
By  \cref{prop:Num_Pairs_Constant}, $|\pi\inv(p')|=N$, and so by the pigeonhole principle, there is a $V_j$ such that $V_j \cap \pi\inv(p') = \emptyset$. 
But this contradicts $p' \in U =\bigcap \pi(V_i) $, so $x=y$ proving injectivity.
For surjectivity, let $(p',i) \in U \times [N]$ be given. 
Again, $V_i \cap \pi\inv(p') \neq \emptyset$ since $p' \in U = \bigcap \pi(V_i)$, so let $x \in V_i \cap \pi\inv(p')$. 
Then $\phi(x) = (p',i)$ as required for surjectivity. 

Finally, we need to ensure that $\phi$ is continuous with a continuous inverse. 
Since $U \times [N]$ has the product topology and $[N]$ has the discrete topology, an open set is of the form $\tilde U \times I$  where $\tilde U \subset U$ is open in $B$ and $I\subset[N]$ is any set. 
Then 
\[ 
\phi \inv (\tilde U \times I) = \displaystyle\bigcup_{i\in I} (\pi\inv(\tilde U) \cap V_i).
\]
Since $\pi$ is continuous, $\pi\inv(\tilde U)$ is open in $E$. 
Thus $\phi \inv (\tilde U \times I) $ consists of finite intersections and unions of open sets, and thus $\phi$ is continuous. 
Showing that $\phi$ has a continuous inverse is equivalent to showing that it is an open map. 
Let $W \subset \pi \inv(U)$ be open in $E$. 
As $\pi$ is an open map (\cref{rem:pi_open}), $\pi(W)$ is open in $B$.  
So, 
\[\phi(W) = \displaystyle\bigcup_{i \in [N]} \left( \pi(W \cap V_i) \times \{i\}\right)\]
is open, thus $\phi$ is an open map and so $\phi$ has a continuous inverse.
Taken together, $\phi$ is a homeomorphism, so $\cE \setminus S$ is a fibered A-canopy.

Assume $S' \subsetneq S$ exists where $\cE \setminus S'$ is fibered. 
Then there is a $p \in S \setminus S'$ for which $\pi \inv(p)$ has a non-Hausdorff point $x$. 
In particular, this means that for the neighborhood $U \ni p$ given by the fiber bundle, $\pi \inv(U)$ contains $x$ and thus is not a Hausdorff space. 
However, $U \subset B$ was assumed to be Hausdorff. 
But then the given map $\phi\colon\pi \inv(U) \to U \times [N]$ is a homeomorphism between a non-Hausdorff space and a Hausdorff space, which is a contradiction. 
Thus, $S$ is a minimal set for which $\cE \setminus S$ is fibered. 
\end{proof}

\section{D-Canopies}
\label{sec:D-Canopy}

Next, we give our definition of a diminished canopy. 
This construction is the closest to what one familiar with TDA would likely be thinking of in terms of parameterized diagrams since we will encode only the off-diagonal point behavior rather than the full augmented persistence diagrams.

\subsection{Definition of D-canopy}
\label{ssec:d-canopy-defn}

Given the augmented canopy information, we can drop the contents on the diagonal to get a structure which we call a diminished canopy. 
However, for these D-canopies, we want to enforce the behavior of points entering and leaving the diagonal. 
This leads to the following construction.

\begin{definition}
\label{def:D-canopy}
A \define{diminished canopy} (or simply \define{D-canopy}) is a collection $\cE' = (E',B,\pi', \eta')$ such that $(E',B,\pi')$ is a bundle, $\eta'\colon  E' \to H'$,  and $(\pi')\inv(p)$ is a finite discrete set for any $p \in B$.
Further, every $p \in B$ has an open neighborhood $U$ such that  
either $\pi \inv(U) = \emptyset$, or  there is a homeomorphism $\phi$
such that 
\begin{equation}
\label{eq:ball_half_space_homeomorphism}
\begin{tikzcd}
 (\pi') \inv(U) \ar[r, "\phi"] \ar[d,"\pi'"'] & \displaystyle\coprod_{i=1}^k U_i \ar[dl, "\iota"] \\ 
U & 
\end{tikzcd}
\end{equation}
commutes, where each $U_i \subset U$, and $\iota$ sends $x\mapsto x$.
In addition,  
\[
\eta'\left(\varphi\inv(\partial(U_{i}) \setminus \partial(U))\right)= \{\Delta\}
\]
where $\partial U$ and $\partial U_i$ are the boundary of each set as a subset of $B$.

Given $\cE' = (E',B,\pi',\eta')$ with $\pi'\colon E' \to B$ and $\eta'\colon B \to H'$ and a set $S \subset B$ for which 
\[\cE' \setminus S \coloneqq (E' \setminus{{\pi'}\inv(S)},B \setminus S,\pi', \eta')\]
is a D-canopy, then we call $\cE'$ an \define{approximate D-canopy}.
\end{definition}

\begin{figure}[h]
    \centering
    \resizebox{.5\linewidth}{!}{
\begin{tikzpicture}
\def\firstball{lizRed} 
\def\firstarc{lizRed} 
\def\secondball{lizDkPur} 
\def\secondarc{lizDkPur} 
\def\thirdball{lizBlue} 
\def\thirdarc{lizBlue} 
\node[scale=2] (doubleLoopLabel) at (2.3,6.8) {$E$};
\node (doubleLoop) at (0,5) {
\begin{tikzpicture}
\begin{axis}[
    hide axis,
    view={-40}{45}, 
    width=8cm,
    height=8cm,
]
\def\R{5} 
\def\w{1.0} 
\addplot3[
    samples y=0,
    thick,
    samples=60,
    domain=0:360,
    variable=x,
    black, 
] (
    {(\R + \w*cos(x/2)) * cos(x)},
    {(\R + \w*cos(x/2)) * sin(x)},
    {\w*sin(x/2)}
);
\addplot3[
    samples y=0,
    thick,
    samples=70,
    domain=0:200,
    variable=x,
    black, 
] (
    {(\R - \w*cos(x/2)) * cos(x)},
    {(\R - \w*cos(x/2)) * sin(x)},
    {-\w*sin(x/2)}
);

\addplot3[
    samples y=0,
    line width=3pt,
    samples=60,
    domain=210:233,
    variable=x,
    \firstarc, 
] (
    {(\R + \w*cos(x/2)) * cos(x)},
    {(\R + \w*cos(x/2)) * sin(x)},
    {\w*sin(x/2)}
);

\addplot3[
    samples y=0,
    line width=3pt,
    samples=70,
    domain=0:30,
    variable=x,
    \secondarc, 
] (
    {(\R - \w*cos(x/2)) * cos(x)},
    {(\R - \w*cos(x/2)) * sin(x)},
    {-\w*sin(x/2)}
);
\addplot3[
    samples y=0,
    line width=3pt,
    samples=60,
    domain=20:37,
    variable=x,
    \secondarc, 
] (
    {(\R + \w*cos(x/2)) * cos(x)},
    {(\R + \w*cos(x/2)) * sin(x)},
    {\w*sin(x/2)}
);

\addplot3[
    samples y=0,
    line width=3pt,
    samples=70,
    domain=147:170,
    variable=x,
    \thirdarc, 
] (
    {(\R - \w*cos(x/2)) * cos(x)},
    {(\R - \w*cos(x/2)) * sin(x)},
    {-\w*sin(x/2)}
);
\addplot3[
    samples y=0,
    line width=3pt,
    samples=60,
    domain=135:165,
    variable=x,
    \thirdarc, 
] (
    {(\R + \w*cos(x/2)) * cos(x)},
    {(\R + \w*cos(x/2)) * sin(x)},
    {\w*sin(x/2)}
);
\end{axis}
\end{tikzpicture}
};

\node[scale=2] (baseSpaceLabel) at (2.7,-.5) {$B$};
\node (baseSpace) at (0,0) {
\begin{tikzpicture}
\begin{axis}[
    hide axis,
    view={40}{40},
    width=8cm,
]
\addplot3 [
    samples=50,
    domain=0:360,
    samples y=0,
    thick
] (
    {cos(x)},
    {sin(x)},
    {0});

\addplot3 [
    samples=50,
    domain=300:320,
    samples y=0,
    line width=3pt,
    \firstarc
] (
    {cos(x)},
    {sin(x)},
    {0});
\addplot3 [
    samples=50,
    domain=100:120,
    samples y=0,
    line width=3pt,
    \secondarc
] (
    {cos(x)},
    {sin(x)},
    {0});
\addplot3 [
    samples=50,
    domain=205:235,
    samples y=0,
    line width=3pt,
    \thirdarc
] (
    {cos(x)},
    {sin(x)},
    {0});

\end{axis}
\end{tikzpicture}
};
\node[scale=2] (DgmLabel) at (8,2) {$H$};
\node (Dgm) at (7,3) {
\adjustbox{scale=.7}{
\begin{tikzpicture}
\tikzset{
dot/.style = {circle, fill, minimum size=#1,inner sep=0pt, outer sep=0pt},
dot/.default = 6pt 
}
\def\xmax{6}
\def\ymax{6}
\draw[->, thick=2] (0,0) -- (\xmax + 0.5, 0);
\draw[->, thick=2] (0,0) -- (0, \ymax + 0.5);
\draw[dashed] (0,5.5) -- (\xmax + 0.5, 5.5);
\node[below left,scale=1.8] at (\xmax + 0.5,5.5) {$\infty$};
\draw[thick=2,dashed] (0,0) -- (\xmax, \ymax);
 \draw[black,line width=1.5pt]
    (2,2)
    .. controls (0.9,1.6) and (0.9,2.6) .. (1.4,3.0)
    .. controls (1.9,3.5) and (1.4,4.5) .. (2.2,4.8)
    .. controls (2.8,5.0) and (3.2,4.6) .. (2.9,4.2)
    .. controls (2.7,3.9) and (2.9,3.6) .. (3.3,3.6)
    .. controls (3.7,3.8) and (4.3,4.5) .. (3.9,4.8)
    .. controls (3.5,5.1) and (4.2,5.4) .. (4.6,5.2)
    .. controls (4.8,5.0) and (4.8,5.0) .. (4.8,4.8)
;
    
\node[fill=\secondball,draw=\secondball,regular polygon, regular polygon sides=3, scale=0.5] at  (1.1,2.2) {};
\node[fill=\thirdball,draw=\thirdball,aspect=1, scale=0.8] at (1.37,3) {};
\draw[fill=\firstball,draw=\firstball] (2,4.7) circle[radius=3pt];
\node[fill=\secondball,draw=\secondball,regular polygon, regular polygon sides=3, scale=0.5] at  (2.87,3.8) {};
\node[fill=\thirdball,draw=\thirdball,aspect=1, scale=0.8] at (4.5,5.24) {};
\end{tikzpicture}
}
};

\draw[->, thick=1.5] (doubleLoop) to [out=-52,in=35] (baseSpace);
\node[right,scale=2] (gammaBLabel) at (1.9,2) {$\pi$};
\draw[->, thick=1.5] (doubleLoop) to [out=10,in=145] (Dgm);
\node[right,scale=2] (gammaHLabel) at (2.9,5.9) {$\eta$};

\draw[fill=\firstball,draw=\firstball] (0,-1) circle[radius=3pt];
\node[fill=\secondball,draw=\secondball,regular polygon, regular polygon sides=3, scale=0.5] at (0.8,0.95) {};
\node[fill=\thirdball,draw=\thirdball,aspect=1, scale=0.8] at (-2.25,0) {};
\draw[fill=\firstball,draw=\firstball] (0,4.95) circle[radius=3pt];
\node[fill=\secondball,draw=\secondball,regular polygon, regular polygon sides=3, scale=0.45] at  (1.1,5.7) {};
\node[fill=\secondball,draw=\secondball,regular polygon, regular polygon sides=3, scale=0.45] at  (1.1,6.7) {};
\node[fill=\thirdball,draw=\thirdball,aspect=1, scale=0.8] at (-1.7,3.6) {};
\node[fill=\thirdball,draw=\thirdball,aspect=1, scale=0.8] at (-2,6) {};

\node[fill=black,draw=\thirdball,aspect=1, scale=0.8] at (2.15,6) {};
\node[aspect=1, scale=1.5] at (2.6,6) {$e_1$};

\node[fill=black,draw=\thirdball,aspect=1, scale=0.8] at (-.8,2.85) {};
\node[aspect=1, scale=1.5] at (-.3,2.85) {$e_2$};

\end{tikzpicture}
}
    \caption{A cartoon of a D-canopy (\cref{def:D-canopy}), modified from \cref{fig:A-canopy-example}. }
    \label{fig:D-canopy-example}
\end{figure}
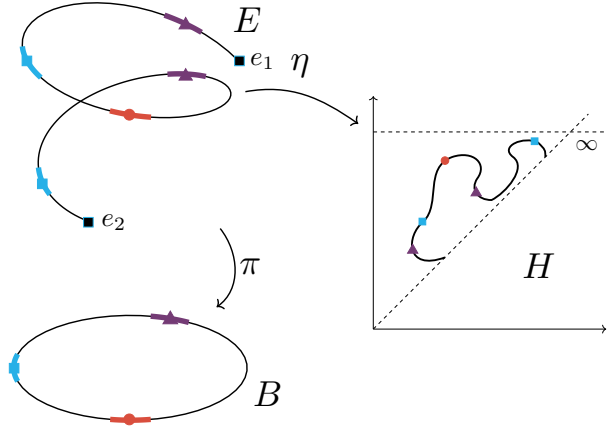

Consider \cref{fig:D-canopy-example}, with a cartoon example of a D-canopy, analogous to the A-canopy of \cref{fig:A-canopy-example}. 
The portion of the image of $E$ in $H$ that was contained in the diagonal has been deleted to construct $E'$. 
Note that we still retain the boundary points of $E'$. 
In particular, the points $e_1$ and $e_2$ are the closure of the set for which $E$ maps to an off-diagonal point; here $\eta'(e_1)=\eta'(e_2) = \Delta$.

The idea behind this construction is to behave essentially like a fiber bundle where possible. 
However, the inverse image under $\pi'$ of an open set $U$ might not give a full homeomorphic copy of $U$. 
This happens when a portion of the points in the relevant diagrams goes into the diagonal. 
Ensuring that the portions of $U$ that are missing in the inverse image correspond only to things entering the diagonal is enforced by the requirement $\eta'\left(\varphi\inv(\partial(U_{i}) \setminus \partial(U))\right)= \{\Delta\}$. 
In essence, this means that points that are in the boundary of $U_i$ that were not already present in $U$ are required to go to the diagonal.

Our main goal of this section is to show that any parameterized filtration function has a D-canopy. 
To relate a potential D-canopy to a fibered filtration function $g$, we have the following definition, which parallels \cref{def:A-induced}.

\begin{definition}\label{def:D-induced}
Given $G\colon B \to \DgmDsp$, the D-canopy $\cE' = (E',B,\pi', \eta')$ \define{is an induced D-canopy of $G$} if for every $p \in B$, 
\[G_D(p) = \eta'((\pi')\inv(p)) \setminus \{ \Delta\} \]
considered as a multiset. 
The \define{induced $B$-parameterized diagram} of a D-canopy $\cE' = (E',B',\pi',\eta')$ is 
\[ 
\begin{matrix}
G_D^{\cE'} \colon & B & \longrightarrow & \DgmDsp\\ 
& p & \longmapsto & \eta'((\pi')\inv(p)) \setminus \{  \Delta\}.
\end{matrix}
\]
\end{definition}

First, note the care taken with the diagonal. 
As we defined D-diagrams to have only strictly off-diagonal points, the example points $e_1$, $e_2$ from \cref{fig:D-canopy-example} would not map to something explicitly in the diagrams over their respective $p \in B$. 
This is remedied by deleting $\{\Delta\}$ from consideration in each case. 
Also, note that in a similar fashion to \cref{rem:MetaDiagram-A-canopies}, we always have an induced $B$-parameterized diagram from a D-canopy; 
however, the existence of a D-canopy for an input fibered filtration function is the focus of \cref{thm:D-bundle-decomp}.

\subsection{Structure theorem for D-canopies}

Now, we can show that any fibered filtration function has a D-canopy so long as we get rid of the structural points. 

\begin{definition}
\label{def:Structural-Coincidental_D}
Fix a Hausdorff space $B$, an D-canopy $\cE' = (E',B,\pi', \eta')$ and the induced  $G_{\cE'}\colon B \to \DgmDsp$, $p \mapsto \eta'((\pi') \inv(p))$ (\cref{def:D-induced}). 
A point $(b,d) \in G_{\cE'}(p)$ with  $|(\eta') \inv(b,d)|\geq 2$ is called \define{coincidental} if all points in 
$(\eta')\inv(b,d) \cap (\pi')\inv(p)$ are Hausdorff in $E$.
Otherwise, $(b,d)$ is called \define{structural}. 
\end{definition} 

\begin{theorem}[D-canopies]
\label{thm:D-bundle-decomp}
Fix a fibered filtration function $g\colon K \times B \to \cD$ with induced diagram function $G\colon B \to \DgmDsp$.
Then $G$ has an induced, approximate D-canopy $\cE'$ with  
\[S = \{p \in B \mid G(p) \text{ has structural points}\}.\] 
\end{theorem}

\begin{proof}
As in \cref{thm:A-bundle-decomp}, the proof is constructive where we build $\cE' = (E',B,\pi', \eta')$. 
From \cref{thm:A-bundle-decomp}, there is an approximately fibered A-canopy $\cE = (E,B, \pi, N,\eta)$ for $G_A\colon B \to \DgmAsp$, which is to say that $\cE \setminus S$ is a fibered A-canopy. 
Recall from the proof of \cref{thm:A-bundle-decomp} that we have fixed a chosen pairing $P_{\prec_p}$ for each $p \in B$ for a compatible order $\prec_p$ for  the function $g_p$.
This was used to build the space
$E = \{(p,\sigma, \tau) \mid p \in B, (\sigma,\tau) \in  P_{\prec_p} \}$ with topology basis given by  
\begin{equation*}
U_e = \{ e' = (p',\sigma',\tau') \mid \exists\,  \Gamma \colon e \rightsquigarrow e' \text{ that projects to }U \} \subseteq E
\end{equation*}
for all $e = (p,\sigma,\tau) \in E$ and neighborhood $U \ni p$.
Let 
\begin{equation*}
\tilde{E} = 
\{(p,\sigma, \tau) 
\mid p \in B, (\sigma,\tau) \in  P_{\prec_p}, g_p(\sigma) \neq g_p(\tau) \} \subseteq E
\end{equation*}
and let $E' = \mathrm{cl}(\tilde{E}) \subset E$ be the closure of $\tilde E$ in $E$. 
We give $E'$ the subspace topology from $E$.

The map $\pi'\colon E' \to B$ is given by the restriction of $\pi\colon E \to B$, the projection on the first factor. 
In order to construct the map $\eta' \colon E' \to H'$ (recall \cref{def:AllTheHs} for the forms of $H$), 
we note that the map $ \eta\colon E \to H$ sending $(p,\sigma,\tau) \to (f_p(\sigma), f_p(\tau))$ can be restricted to $E' \subseteq E$ and then post-composed with the quotient map $H \to (H / \sim) = H' = H\cup \{\Delta\}$ to get a map $\eta'\colon E' \to H'$. 
In essence, this map $\eta'$ is the same as $\eta$ when the diagram point is in $\mathring{H}$, and for any $e = (p,\sigma,\tau) \in E'$ with $g_p(\sigma) = g_p(\tau)$, we have $\eta'(e) = \Delta$. 
We will show that $\cE' \setminus S = (E' \setminus (\pi')\inv S, B \setminus S, \pi', \eta')$ is a D-canopy.

Fix a $p \in B\setminus S$.
Because $\cE \setminus S$ is a fibered A-canopy, we can assume that $U \subset B \setminus S$ and there is a homeomorphism $\phi$ such that the diagram
\begin{equation}
\begin{tikzcd}
 \pi \inv(U) \ar[r, "\phi"] \ar[d,"\pi"'] & U \times [N] \ar[dl, "{\text{proj}_1}"] \\ 
U & 
\end{tikzcd}
\end{equation}
commutes.
For each $i \in [N]$, denote 
$V_i = \phi \inv(U \times \{i\} )$ and $V_i' = V_i \cap E'$,  noting that $(\pi')\inv(U) = \coprod V_i'$. 
Define $U_i = \pi( V_i')$.
Then we have a restriction of $\phi$ to $(\pi')\inv(U)$,
\begin{equation}
\begin{tikzcd}
 (\pi') \inv(U) \ar[r, "\phi"] \ar[d,"\pi'"'] & \displaystyle\coprod_{i=1}^k U_i \ar[dl, "\iota"] \\ 
U & 
\end{tikzcd}
\end{equation}
where we abuse notation and write the restriction also as $\phi$. 
By construction, the restricted $\phi$ is still a homeomorphism.

It remains to show that $\eta'\left(\phi\inv(\partial(U_{i}) \setminus \partial(U)) \right)= \{\Delta\}$. 
Let $x \in \partial(U_{i}) \setminus \partial(U)$ and let $e = \phi\inv(x) \in V_i$, writing $e=(p_e,\sigma_e,\tau_e)$. 
If $e \in \cl(\tilde E) \setminus \tilde E$, then it must be that $g_{p_e}(\sigma_e) = g_{p_e}(\tau_e)$ and thus $\eta'(e) = \Delta$, so we are done.

Assume instead $e \in \tilde E \cap V_i$. 
Note that $\tilde E$ is open in $E$ since $E \setminus \tilde E = \eta \inv(\{(x,x) \in H \}$ is the inverse image of a closed set under a continuous map, and therefore closed.
So there is an open $W$ with $e \in W \subset \tilde E \cap V_i \subset V_i'$.
Thus $\phi(W) \subset U_i$ and contains $x = \phi(e)$, and is open because $\phi$ is a homeomorphism. 
However, since $x \in \partial U_i \setminus \partial U$ and $\phi(W) \ni x$ is an open neighborhood, there is a point $q \in \phi(W) \cap (U \setminus U_i)$ which is a contradiction to the assertion that $\phi(W) \subset U_i$, completing the proof.
\end{proof}

We conclude with an immediate corollary of the same result from the A-canopies that this construction is unique up to isomorphism which we state for completeness here.

\begin{definition}
Two D-canopies 
$\cE'_1 = (E'_1,B,\pi'_1, \eta'_1)$
and
$\cE'_2 = (E'_2,B,\pi'_2, \eta'_2)$
 are isomorphic if there is a homemorphism $\phi\colon E_1' \to E_2'$ such that the diagrams
\begin{equation*}
\begin{tikzcd}
E_1' \ar[rr,"\phi'"] \ar[dr, "\pi_1'"'] && E_2' \ar[dl,"\pi_2'"] 
& 
E_1' \ar[rr,"\phi'"] \ar[dr, "\eta_1'"'] && E_2' \ar[dl,"\eta_2'"]\\
&B &&& H
\end{tikzcd}
\end{equation*}
commute.
\end{definition}

\begin{theorem}
\label{thm:D-bundle-unique}
The D-canopy of \cref{thm:D-bundle-decomp} is unique up to isomorphism.
\end{theorem}
\begin{proof}
This is an immediate result from combining \cref{thm:A-bundle-unique} with the proof of \cref{thm:D-bundle-decomp} since it is directly constructed from the given A-canopy.
\end{proof}

\begin{examplet}
\label{ex:pht}
Consider the example of \cref{fig:D-CanopyPHT}. 
Here we have an embedded graph $K$ with vertices $a$, $b$, and $c$. 
The directional transform gives a fibered filtration function where $f_v(\omega) = \langle v, \omega\rangle$ for $\omega \in \S^1$, and $v$ any of the vertices $\{a,b,c\}$. 
The function on edges is from the lower-star filtration, namely $f_{uv}(\omega) = \max \{f_u(\omega), f_v(\omega) \}$.
Choices of pairings are shown for each portion of the circle at left. 
In this example, the full A-canopy has $E$ given by a disjoint union of three circles.
The first blue pair in each list gives the infinite class. 
Since this is never on the diagonal, the entire circle shows up in both $E$ and $E'$, with image on the infinity line shown in the far right $H$. 
The second, purple pairs are always equivalent to each other because the pairings before and after a critical change of direction are filtration equivalent, resulting in a second circle in $E$. 
In this case, only the pairs in the upper left quadrant of the circle result in off-diagonal points in the diagram.
This is shown where only the top left quadrant of the second circle in $E$ is included in $E'$; we mark the portion in $E'$ with a solid line and the portion in $E\setminus E'$ with a dashed line. 
Its image comes out of the diagonal and reenters on the same path, so the image is drawn as a single curved purple line in the far right $H$. 
The third pair in each collection, written in yellow, never leaves the diagonal. 
So there is a third circle in $E$ which is not included in $E'$.
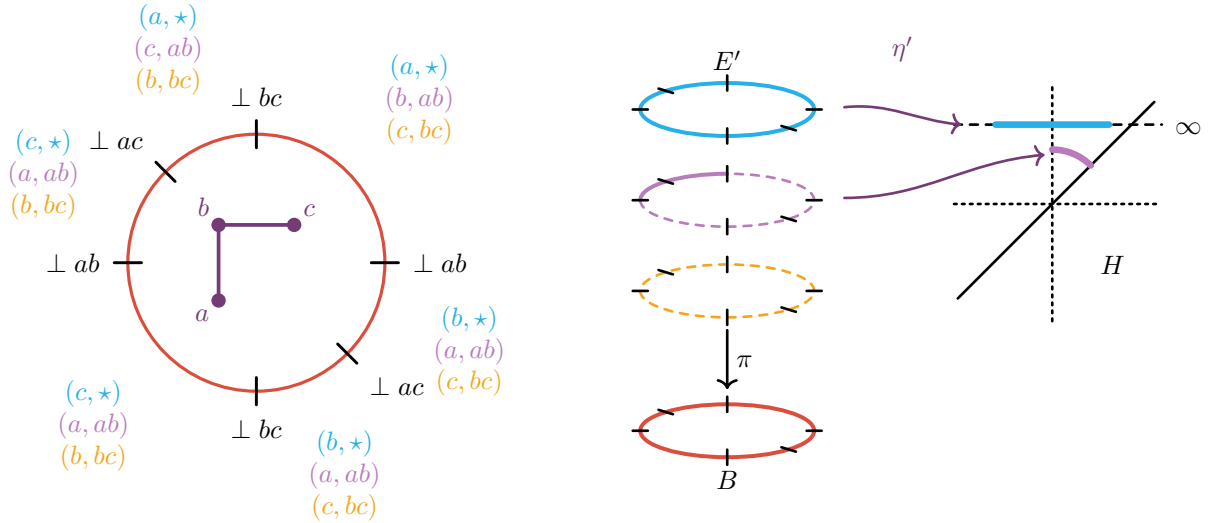
\begin{figure}[h!]
    \centering
    \begin{minipage}
    {0.5\textwidth}
    \begin{tikzpicture}[scale=1,line cap = round]
	\def\R{1.7}
	\def\tick{0.18}
	\def\annR{1.35}

	\draw[lizRed, line width=1.2pt] (0,0) circle[radius=\R];

	\draw[lizDkPur, line width=1.4pt, line cap=round]
		(-.5,-.5) -- (-.5,.5) -- (.5,.5);
	\fill[lizDkPur] (-.5,-0.5) circle (2.7pt) node[below left] {$a$};
	\fill[lizDkPur] (-.5,0.5) circle (2.7pt) node[above left] {$b$};
	\fill[lizDkPur] (0.5,0.5) circle (2.7pt) node[above right] {$c$};

	\draw[black, line width=1.1pt] (0:{\R-\tick}) -- (0:{\R+\tick});
	\node[black, anchor=west] at (0:{\R+0.25}) {$\perp ab$};

	\draw[black, line width=1.1pt] (90:{\R-\tick}) -- (90:{\R+\tick});
	\node[black, anchor=south] at (90:{\R+0.25}) {$\perp bc$};

	\draw[black, line width=1.1pt] (135:{\R-\tick}) -- (135:{\R+\tick});
	\node[black, anchor=south east] at (135:{\R+0.25}) {$\perp ac$};

	\draw[black, line width=1.1pt] (180:{\R-\tick}) -- (180:{\R+\tick});
	\node[black, anchor=east] at (180:{\R+0.25}) {$\perp ab$};

	\draw[black, line width=1.1pt] (270:{\R-\tick}) -- (270:{\R+\tick});
	\node[black, anchor=north] at (270:{\R+0.25}) {$\perp bc$};

	\draw[black, line width=1.1pt] (315:{\R-\tick}) -- (315:{\R+\tick});
	\node[black, anchor=north west] at (315:{\R+0.25}) {$\perp ac$};

	\node[align=center] at (45:{\R+\annR}) {${\color{lizBlue}(a,\star)}$\\${\color{lizLtPur}(b,ab)}$\\${\color{lizYellow}(c,bc)}$};
	\node[align=center] at (112.5:{\R+\annR}) {${\color{lizBlue}(a,\star)}$\\${\color{lizLtPur}(c,ab)}$\\${\color{lizYellow}(b,bc)}$};
	\node[align=center] at (157.5:{\R+\annR}) {${\color{lizBlue}(c,\star)}$\\${\color{lizLtPur}(a,ab)}$\\${\color{lizYellow}(b,bc)}$};
	\node[align=center] at (225:{\R+\annR}) {${\color{lizBlue}(c,\star)}$\\${\color{lizLtPur}(a,ab)}$\\${\color{lizYellow}(b,bc)}$};
	\node[align=center] at (292.5:{\R+\annR}) {${\color{lizBlue}(b,\star)}$\\${\color{lizLtPur}(a,ab)}$\\${\color{lizYellow}(c,bc)}$};
	\node[align=center] at (337.5:{\R+\annR}) {${\color{lizBlue}(b,\star)}$\\${\color{lizLtPur}(a,ab)}$\\${\color{lizYellow}(c,bc)}$};
\end{tikzpicture}
    \end{minipage}
    \begin{minipage}
    {0.4\textwidth}
    \begin{tikzpicture}[scale=1, line cap=round]
	\def\cx{-1.8}
	\def\rx{1.15}
	\def\ry{0.35}

	\def\yt{2.0}
	\def\ym{0.8}
	\def\yb{-0.4}
	\def\ybb{-2.25}

	\node[black] at (\cx,\yt+0.65) {$E'$};

	\draw[lizBlue, line width=1.5pt] (\cx,\yt) ellipse [x radius=\rx, y radius=\ry];

	\draw[lizLtPur, line width=1.0pt, dashed] (\cx,\ym) ellipse [x radius=\rx, y radius=\ry];
	\draw[lizLtPur, line width=1.5pt] ($(\cx,\ym)+(90:\ry)$)
		arc[start angle=90, end angle=180, x radius=\rx, y radius=\ry];

	\draw[lizYellow, line width=1.0pt, dashed] (\cx,\yb) ellipse [x radius=\rx, y radius=\ry];

	\draw[lizRed, line width=1.5pt] (\cx,\ybb) ellipse [x radius=\rx, y radius=\ry];
	\node[black] at (\cx,\ybb-0.65) {$B$};

	\draw[black, line width=1.0pt, ->] (\cx,\yb-0.52) -- (\cx,\ybb+0.55);
	\node[black] at (\cx+0.23,{(\yb+\ybb)/2}) {$\pi$};

	\foreach \yy in {\yt,\ym,\yb,\ybb}{
		\draw[black, line width=1pt] (\cx+\rx+0.10,\yy) -- (\cx+\rx-0.10,\yy);
		\draw[black, line width=1pt] (\cx-\rx-0.10,\yy) -- (\cx-\rx+0.10,\yy);
		\draw[black, line width=1pt] (\cx,\yy+\ry+0.10) -- (\cx,\yy+\ry-0.10);
		\draw[black, line width=1pt] (\cx,\yy-\ry-0.10) -- (\cx,\yy-\ry+0.10);

		\pgfmathsetmacro{\xL}{\rx*cos(135)}
		\pgfmathsetmacro{\yL}{\ry*sin(135)}
		\pgfmathsetmacro{\nL}{sqrt((\xL)^2 + (\yL)^2)}
		\draw[black, line width=1pt]
			({\cx + \xL + 0.10*(\xL/\nL)}, {\yy + \yL + 0.10*(\yL/\nL)})
			--
			({\cx + \xL - 0.10*(\xL/\nL)}, {\yy + \yL - 0.10*(\yL/\nL)});

		\pgfmathsetmacro{\xR}{\rx*cos(315)}
		\pgfmathsetmacro{\yR}{\ry*sin(315)}
		\pgfmathsetmacro{\nR}{sqrt((\xR)^2 + (\yR)^2)}
		\draw[black, line width=1pt]
			({\cx + \xR + 0.10*(\xR/\nR)}, {\yy + \yR + 0.10*(\yR/\nR)})
			--
			({\cx + \xR - 0.10*(\xR/\nR)}, {\yy + \yR - 0.10*(\yR/\nR)});
	}

	\begin{scope}[shift={(2.5,0.75)}]
		\draw[black, line width=1pt,dotted] (-1.30,0) -- (1.4,0);
		\draw[black, line width=1pt,dotted] (0,-1.55) -- (0,1.55);
		\draw[black, line width=1pt] (-1.25,-1.25) -- (1.35,1.35);

		\draw[black, line width=1.0pt, dashed] (-1.30,1.05) -- (1.45,1.05);
		\draw[lizBlue, line width=2.50pt] (-0.75,1.05) -- (0.75,1.05);
		\node[black] at (1.8,1) {$\infty$};

		\draw[lizLtPur, line width=2.5pt] (0,0.72) arc[start angle=90, end angle=45, radius=0.72];

		\node[black] at (0.8,-0.8) {$H$};
	\end{scope}

	\node[lizDkPur] at (0.52,2.8) {$\eta'$};

	\draw[lizDkPur, line width=1.0pt, ->] (-0.25,\yt+0.03) to[out=6, in=175] (1.30,1.8);

	\draw[lizDkPur, line width=1.0pt, ->] (-0.25,\ym+0.03) to[out=0, in=185] (2.4,1.4);
\end{tikzpicture}
    \end{minipage}
    \caption{The D-canopy (right) for an input fibered filtration function given by the PHT of the embedded graph at left. See \cref{ex:pht} for a full description.}
    \label{fig:D-CanopyPHT}
\end{figure}
\end{examplet}

\section{Implications}
\label{sec:implications}

In this section, we describe some immediate implications of the constructions above. 

\subsection{Defining a vine}

First, a goal of doing this work was to provide a mathematically rigorous definition of a vine, which is often mentioned in the context of vineyards in the literature, but always has the caveat (explicit or implicit) that it is only well defined away from points of multiplicity. 
For example, \cite{CohenSteiner2006} notes ``it is not clear whether arbitrary vineyards can be decomposed into well-defined, individual vines."
The construction of the canopy implies an answer to this question; there is no well-defined decomposition into simple paths because, at structural points, either choice of pairing is equally correct. 
In particular, \cref{ex:L_fullA-canopy} shows that while there are four 1-dimensional pieces, they are inextricably connected by the non-Hausdorff point and thus cannot be separated across that location.
Thus, we propose the following definition for a vine. 

\begin{definition}
\label{defn:vine}
    An \define{augmented (resp.~diminished) \emph{vine}}, also called \define{A-vine} (resp., \define{D-}), is a canopy $\cE$ (resp.~$\cE'$) with connected $E$ (resp.~$E'$).

    Given an $A$-canopy $\cE = (E,B,\pi,N,\eta)$ (resp.~$D$-canopy $\cE' = (E',B, \pi',\eta')$), a \define{vine decomposition} of $\cE$ (resp. $\cE'$) is a collection of $A$-vines $\{\cE_i \}_{i \in I}$ (resp.~$D$-vines $\{\cE'_i \}_{i \in I}$)  such that $E = \coprod_{i \in I} E_i$ (resp.~$E' = \coprod_{i\in I}E_i'$) with additional maps being defined by restrictions. 
    
    If $B$ and all $E_i$ (resp.~$E_i'$) are Hausdorff, we call the vine decomposition \define{regular}. 
\end{definition}

With this definition, we have the following decomposition theorem. 
\begin{proposition}
Any A-canopy (resp.~D-canopy) has a vine decomposition.  
If the B-parameterized diagram for the canopy has no points of multiplicity and $B$ is Hausdorff, then the vine decomposition is regular; and further the vines are fibered for the A-canopy.
\end{proposition}

\begin{proof}
The first statement is immediate from taking $E$ (resp.~$E'$) to be the disjoint union of its connected components. 
The second comes from the fact that the non-Hausdorff points in the A- or D-canopies are a subset of points of multiplicity.
\end{proof}

\subsection{Defining monodromy}

The canopy setting also allows us to give an explicit definition of what we mean by monodromy in $B$-parameterized diagrams.

\begin{definition}
In both A- and D-canopies, $k$-monodromy is defined by a subset $M \subset B$ with $B$ Hausdorff and $M$ homeomorphic to $\S^1$ with $\{C_i\}_{i \in I}$ connected components of $\pi\inv(M)$ which are Hausdorff and 
\[
k = \max_{ C_i}
\sup_{p\in B}
\left|
\pi\inv(p) \cap C_i
\right|.
\]
If $k \geq 2$, we say the canopy has non-trivial monodromy; otherwise we say it has trivial monodromy. 
\end{definition}

Note that in the case of a fibered A-canopy, this definition exactly coincides with a connected component of $\pi \inv(M)$ being a degree $k$ covering map of the circle $M$.
However, this equivalence does not apply in the case of D-canopies. 

\begin{proposition}
\label{prop:no_monodromy_for_fib_spheres}
A fibered A-canopy with $B=\S^d$, for $d\geq 2$, has no non-trivial monodromy. 
\end{proposition}

\begin{proof}
A fiber bundle with discrete fibers is a covering space, and as discussed in \cref{ssec:bundles}, the only covering space for a sphere of dimension at least 2 is the trivial cover. 
So by definition, the maximum number of points over $p$ in any connected component of $\pi\inv(M)$ is 1, and the result follows.
\end{proof}

Note that there is not an equivalent version of this proposition for D-canopies. 
Suppose there is an A-canopy over $\S^d$ with monodromy, but the structural points required by the proposition above are contained in the diagonal. 
Then the resulting D-canopy obtained by deleting the diagonal points would remove these issues.
So we cannot promise that a D-canopy over a higher dimensional sphere does not have monodromy.

\begin{examplet}
\label{ex:AbbyExampleFunction-v2}
Note that by \cref{prop:no_monodromy_for_fib_spheres} and \cref{thm:A-bundle-decomp-generic}, an A-canopy of $\S^d$ for $d\geq 2$ with monodromy implies the existence of a structural point in the canopy. 
However, we can construct an example with a structural point in the A-canopy but no monodromy by modifying \cref{ex:L_2-param}. 
Using the same simplicial complex $L$, we use the same function $f_\sigma$ for all simplices except $c$ and $d$, for which we define setting $r = \sqrt{x^2+y^2}$,
\[
\begin{aligned}
f_c(x,y) &= \sin^2\left(\frac{r}{1+r^2}\right)+4;\,\\
f_d(x,y) &= - \sin^2\left(\frac{r}{1+r^2}\right)+4
\end{aligned}
\]
This function has the property that $f_c(0,0)=f_d(0,0)=4$, and $f_c(x,y)<f_d(x,y)$ for all other points; see \cref{fig:AbbyExamplePosets-v2}. 
In this case, every point which is not on the $x$-axis has a total compatible ordering for $f_{x,y}$: $a \to b \to c \to d$ below the line and $b \to a \to c \to d$ above the line.  
Thus there are unique pairings available in these portions, and crossing the $x$-axis gives a unique combinatorial equivalence, so the result is two distinct sheets in $E$. 
However, at the origin, we have a combinatorial equivalence since we have both pairings $P_1=\{(b,c),(a,d)\}$ and $P_2=\{(a,c),(b,d)\}$, and, for example, $(b,c)$ can be mapped to $(a,c)$ or to $(b,d)$.
Thus the two points over the origin are always in the same open set. 
Hence the origin is still a structural point, but the A-canopy does not have monodromy.
\end{examplet}

We note that the above example still has trivial monodromy even if we perturb the filtration function, so that there is a whole circle of parameters for which $c$ and $d$ have the same filtration value. 

\begin{figure}[h]
\centering
\begin{minipage}{.2\textwidth}
\centering 

\includegraphics[width =.5\textwidth,align=c]{Fig/AbbyExample-justL.pdf}

\includegraphics[width = \textwidth, align = c]{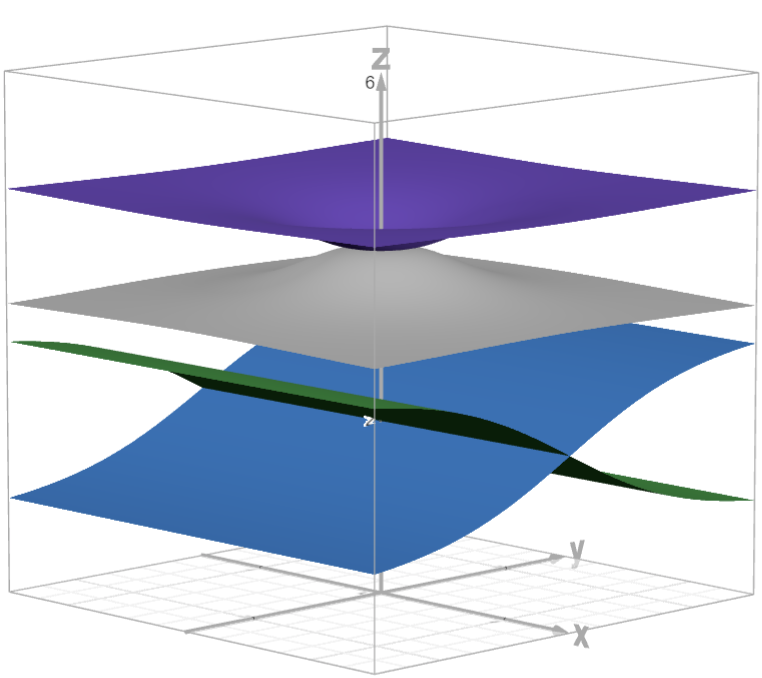}
\end{minipage}
\quad
\includegraphics[width=0.35\linewidth,align=c]{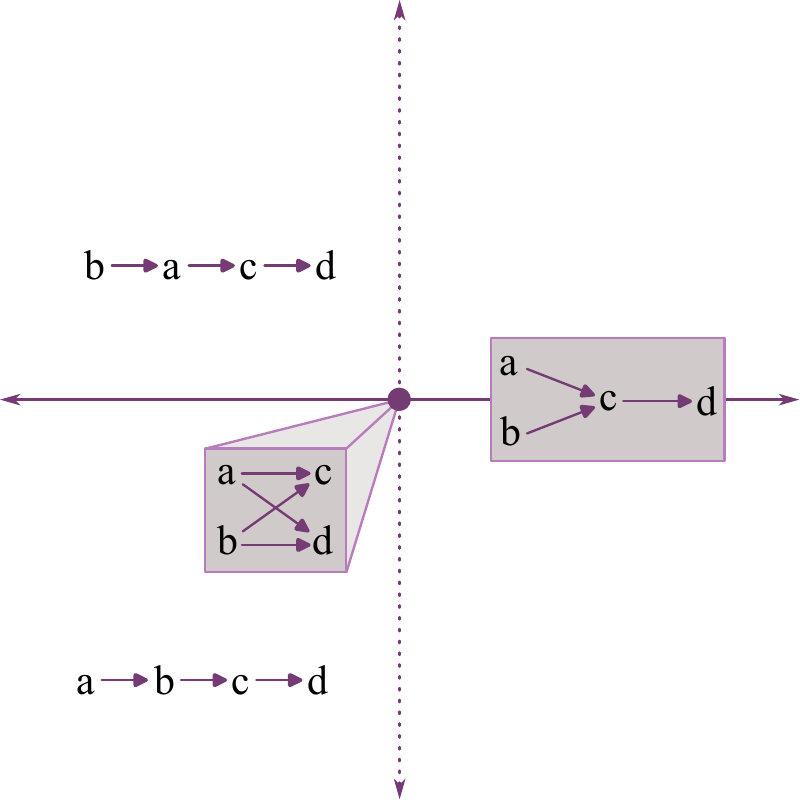}
\includegraphics[width=0.3\linewidth,align=c]{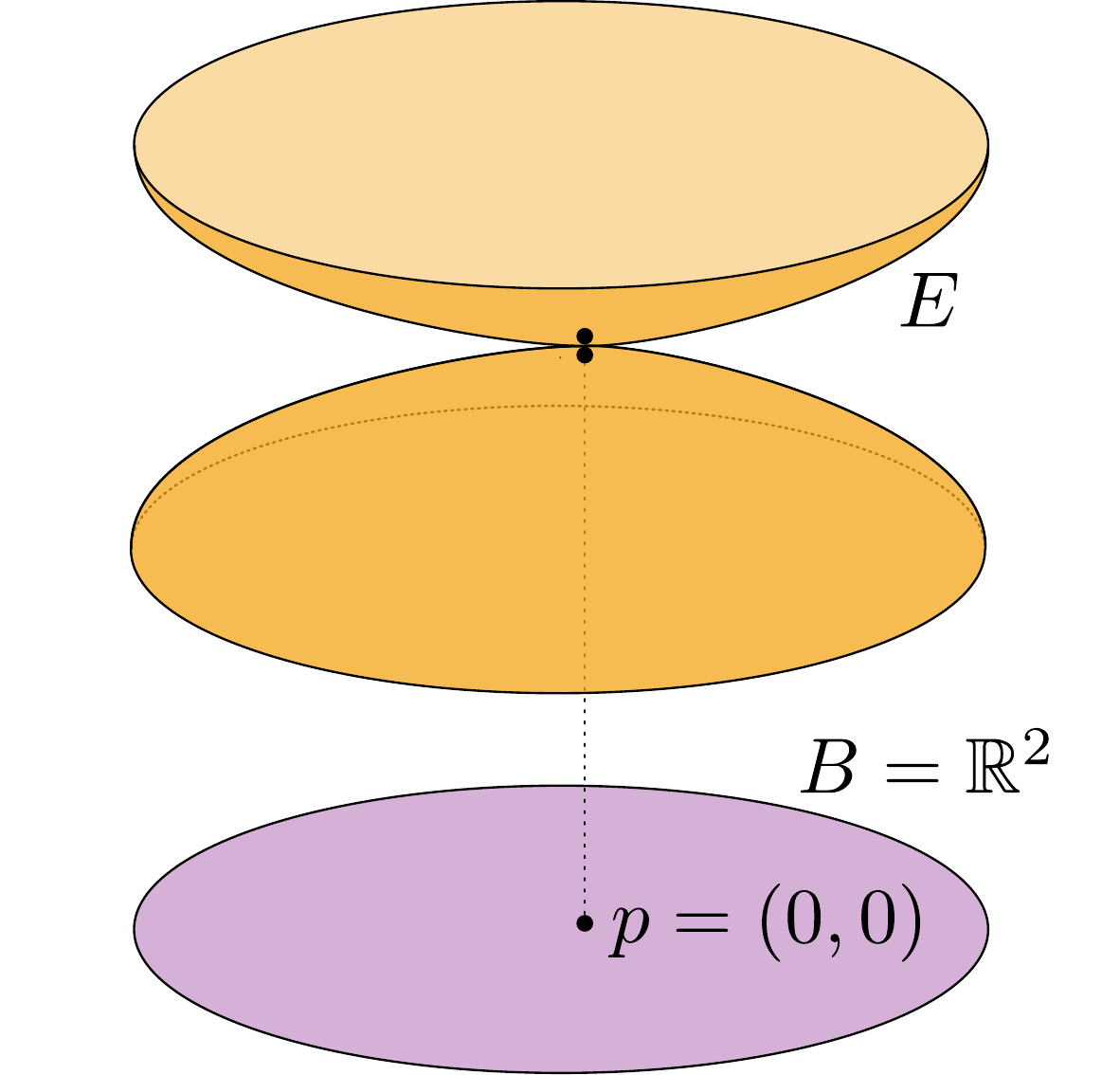}
    \caption{Left is the running example simplicial complex $L$, followed by a graph of the function in \cref{ex:AbbyExampleFunction-v2} with $f_a$ in blue, $f_b$ in green, $f_c$ in purple, and $f_d$ in grey; see \href{https://www.desmos.com/3d/pkbu1iv7lz}{desmos.com/3d/pkbu1iv7lz} for an interactive plot.
    Next is the compatible order for each region, and finally the resulting A-canopy which has a pair of non-Hausdorff points above the origin.
    }
    \label{fig:AbbyExamplePosets-v2}
\end{figure}

\subsection{The canopy of a persistent homology transform}

The persistent homology transform provides a wealth of examples for constructing fibered filtration functions to study. 
Indeed, we have \cref{ex:pht}; but also \cite[Example 2.24]{Arya2024} showing that an embedded spiral has a PHT that exhibits monodromy.
In particular, there is broad interest in understanding what gives rise to monodromy in these settings, as well as what properties of an embedded shape promise no monodromy is possible.
For example, \cite{Arya2024} proves that star-shaped objects in 2D cannot have monodromy, but this is also not a complete characterization and does not extend to higher-dimensional shapes.

For completeness, we recall that the directional transform of an embedded simplicial complex $K \subset \R^d$ is a fibered filtration function 
\begin{equation*}
\begin{matrix}
g \colon & K \times \S^{d-1} & \longrightarrow & \R\\ 
& (\sigma, \omega) & \longmapsto & \max_{v \in \sigma} \langle v, \omega\rangle
\end{matrix}
\end{equation*}
where we abuse notation to write $v$ as a vertex of $K$, and also $v \in \R^n$ is the coordinates of the embedding of $v$. 
Then the PHT of $K \subset \R^d$ is the induced diagram function, traditionally viewed as the D-diagram version 
\begin{equation*}
\begin{matrix}
\mathrm{PHT}(K) \colon & \S^{d-1} & \longrightarrow & \DgmDsp\\ 
& \omega & \longmapsto & \DgmD(g_\omega)
\end{matrix}
\end{equation*}
but of course we can also define the A-diagram version
\begin{equation*}
\begin{matrix}
\mathrm{PHT}_A(K) \colon & \S^{d-1} & \longrightarrow & \DgmAsp\\ 
& \omega & \longmapsto & \DgmA(g_\omega).
\end{matrix}
\end{equation*}
We say that an embedded simplicial complex exhibits monodromy if the A- or D-canopy of $\mathrm{PHT}$ or $\mathrm{PHT_A}$ has monodromy.

As an immediate consequence of our setting, we can use \cref{prop:no_monodromy_for_fib_spheres}
to show that there is a tight relationship between monodromy and the existence of structural points.

\begin{corollary}
Fix a simplicial complex $K$ embedded in $\R^d$ for $d \geq 3$.  
If the  A-canopy of $\mathrm{PHT_A}$ is fibered, then there is no monodromy. 
\end{corollary}

In other words, if the $A$-canopy of a $\mathrm{PHT_A}$ over $\S^d$, for $d\geq 3$, has monodromy, it must have a structural point.

\begin{proof}
By \cref{prop:no_monodromy_for_fib_spheres}, if $\mathrm{PHT}_A$ has a fibered A-canopy, then it consists of a fiber bundle over a sphere of dimension at least 2 and hence must be a trivial bundle. 
Thus, a connected component over any loop in $B$ is a trivial fibered bundle itself and thus does not give non-trivial monodromy.
\end{proof}

\section{Discussion}
\label{sec:Discussion}

In this paper, we provided a new construction for studying parameterized persistence in the setting of a fibered filtration function. 
Specifically, for a function $f \colon B \times K \to \R$, we can track the augmented persistence diagrams as a function $F_A \colon B \to \DgmAsp$ or the diminished persistence diagrams as a function $F_D \colon B \to \DgmDsp$.
We then build an A-canopy by constructing a set $E = \{ (p,\sigma,\tau) \mid (\sigma,\tau) \in P_{\prec_p}\}$ for fixed choices of pairings, and then defining a topology on the space based on trajectories that relate the points algebraically.
Then we can build a D-canopy by taking the A-canopy built and essentially forgetting the points on the diagonal while taking care of boundary conditions. 
We showed that away from particularly bad points of multiplicity in the persistence diagrams, the A-canopy has the structure of a fiber bundle, and showed that the D-canopy has similar restricted structure without being a true fiber bundle.
Indeed, these problematic points arise as non-Hausdorff points in the respective spaces. 

The canopy construction is very versatile, as it allows us to provide a definition of vines that include points with multiplicity and to discuss monodromy, thus either solving or setting the stage to solve some open problems in \cite{Turner2023}, and to answer the long standing question in persistence theory of how to follows the points in the persistence diagram when their paths cross. 
This latter answer may seem unsatisfactory: it is in practice telling us that whenever the choice is non-canonical, we should not choose at all because the features are somewhat morphing into one another. 
Nonetheless, it provides a practical tool to handle these situations rather than having to assume they never occur.

This work opens up a wealth of possible directions for future study. 
We focused on using as few requirements as possible on the fibered filtration function to make the canopies applicable to a wide range of inputs.
In particular, we did not use the stratification studied in \cite{Hickok2026} for general bundles with perturbations; or the stratification built for the PHT specifically given by \cite{Curry2022}. 
We believe additional control on the topological spaces involved, or on the functions themselves, would provide potentially stronger results than were available here. 
It would also be interesting to understand how the resulting constructions that arise from other fibered filtration functions available in the literature, such as those arising from distance transforms, or studying the state space of dynamical systems. 
Further study to understand the resulting constructions for more general input spaces than finite simplicial complexes is also an interesting open direction for the future.

\paragraph{Acknowledgments}
BG and EM are grateful to Erin Chambers, Abigail Hickok, Kate Turner, and Kevin Woytowich for helpful discussions in the course of this work. 
This material is partially based upon work supported by the Swedish Research Council under grant no.~2021-06594 while the authors were in residence at Institut Mittag-Leffler in Djursholm, Sweden during the summer of 2025.
This work was funded in part by the National Science Foundation [CCF-2106578, CCF-2142713, Conference Grant DMS-2500005, EM]; and by a grant from the Simons Foundation [MPS-TSM-00007525, BG].

\printbibliography

\end{document}